\documentclass[11pt,letterpaper,USenglish,thm-restate]{article}

\usepackage{subcaption}

\usepackage{authblk}

\usepackage[dvipsnames]{xcolor}
\usepackage[margin=1in]{geometry}
\usepackage{setspace}
\usepackage{dsfont}
\usepackage{mathrsfs}
\usepackage{amssymb}
\usepackage{amsmath}
\usepackage{amsfonts}
\usepackage{amsthm}
\usepackage{graphicx}
\usepackage{mathtools}
\usepackage{verbatim}
\usepackage{cprotect}
\usepackage{algpseudocode}
\usepackage{pgfplots}
\usepackage[utf8]{inputenc}
\usepackage{listings}
\usepackage{algorithm}
\usepackage{algpseudocode}
\usepackage{xcolor}
\usepackage{hyperref}
\usepackage{enumitem}
\usepackage[capitalize]{cleveref}
\usepackage{array}
\usepackage{booktabs}
\usepackage{mathrsfs}

\usepackage{graphicx,wrapfig,lipsum}

\usepackage{thm-restate}

\usepackage{cite}

\newtheorem{theorem}{Theorem}
\newtheorem{lemma}{Lemma}
\newtheorem{corollary}{Corollary}

\newcommand{\mpmd}{\textsf{MPMD} }
\newcommand{\comp}{\mathscr{C}}
\newcommand{\dec}{\mathfrak{D}}
\newcommand{\rank}{\text{rank}}
\newcommand{\nrank}{\text{nrank}}
\newcommand{\req}{\text{req}}

\newtheorem{claim}{Claim}
\crefname{claim}{Claim}{claims}
\input{macros}

\title{A Deterministic Polylogarithmic Competitive Algorithm for Matching with Delays}
\author[]{Marc Dufay}
\author[]{Roger Wattenhofer}

\affil[]{ETH Z{\"u}rich\\ \{mdufay, wattenhofer\}@ethz.ch}

\date{\vspace{-5ex}}

\begin{document}

\maketitle
%\linenumbers

\begin{abstract}
In the online Min-cost Perfect Matching with Delays (\textsf{MPMD}) problem, $m$ requests in a metric space are submitted at different times by an adversary. The goal is to match all requests while (i) minimizing the sum of the distances between matched pairs as well as (ii) how long each request remained unmatched after it appeared.

While there exist almost optimal algorithms when the metric space is finite and known a priori, this is not the case when the metric space is infinite or unknown. In this latter case, the best known algorithm, due to Azar and Jacob-Fanani, has competitiveness $\mathcal{O}(m^{0.59})$ which is exponentially worse than the best known lower bound of $\Omega(\log m /  \log \log m)$ by Ashlagi et al.

We present a $\mathcal{O}(\log^5 m)$-competitive algorithm for the \textsf{MPMD} problem. This algorithm is deterministic and does not need to know the metric space or $m$ in advance. This is an exponential improvement over previous results and only a polylogarithmic factor away from the lower bound.
\end{abstract}
\section{Introduction}

We consider an online two-player game like online chess or racing. Players arrive one by one and must be paired to play games against each other. A well-designed matchmaking system needs to optimize two conflicting criteria: first, we need to minimize the players' waiting time to maintain engagement, second, we need to ensure that players are satisfied with whom they are matched with. This satisfaction might depend on multiple factors, e.g., the skill difference with their opponent or their geographical distance, which affects game latency.

%We consider an online two-player game like chess or an online racing game. Players arrive one by one and get matched to play games together. A proper matchmaking system would want to minimize both the frustration of playing against someone out of your skill level, and the growing impatience of players waiting to be matched.

This trade-off was first formalized by Emek et al \cite{STOC:EmeKutWat16} by defining the online Min-cost Perfect Matching with Delay (\textsf{MPMD}) problem. In this problem we are given a metric space and $m$ requests, which are points in the metric space arriving at any time. An algorithm must pair all requests in a way that minimizes (i) the sum of the distances between each pair (space cost) and (ii) how long each request has to wait after arriving before being matched (time cost). 

This problem was first studied in the case where the metric space is finite and known in advance. Most of the proposed algorithms were randomized and relied on metric tree embeddings \cite{STOC:EmeKutWat16,SODA:AzaChiKap17}. We note that this approach has many drawbacks. First, the competitive ratio may be arbitrarily large compared to the actual number of requests. In addition, these algorithms' guarantees only hold in expectation. In particular, these algorithms fail against an adaptive adversary.

%This is an issue as having a finite metric space can restrict ways to describe the skill of a player. Moreover, online matchmaking is typically a setting where one would want an algorithm resilient against an adaptive adversary. Indeed, an adversary can easily choose the skill level and time when a group of players joins. Moreover, it gets immediate feedback once a pair of players gets matched.

A way to work around these issues is to no longer assume knowledge about the metric space and focus on deterministic algorithms. This line of work started with Bienkowski et al \cite{AOA:BKS18} who gave an algorithm with $\mathcal{O}(m^{\log 5.5}) \approx \mathcal{O}(m^{2.46})$ competitiveness. This was later improved to  $\mathcal{O}(m)$ by Bienkowski et al \cite{AOA:BKLS18} and to $\mathcal{O}(\varepsilon^{-1} m^{\log (1.5 + \varepsilon)}) \approx \mathcal{O}(m^{0.59})$ by Azar and Jacob-Fanani \cite{TCS:YJ20}. This is so far the best competitiveness for the general MPMD problem. A lower bound of $\Omega(\log m /  \log \log m)$ also exists by Ashlagi et al \cite{ARCO:AACCGKMWW17}. This lower bound is quite strong as it holds even for a randomized algorithm against an oblivious adversary. 

\subsection{Our contribution}

In their paper \cite{TCS:YJ20}, Azar and Jacob-Fanani ask whether there exists an algorithm with poly-logarithmic competitiveness in the general setting. We answer in the positive by describing a deterministic algorithm with $\mathcal{O}(\log^5 m)$-competitiveness for \mpmd. This represents an exponential improvement over previous results and is only a poly-logarithmic factor away from the best known lower bound of $\Omega(\log m /  \log \log m)$.

\subsection{Related Work} The first algorithm with bounded competitive ratio was given by Emek et al \cite{STOC:EmeKutWat16} and focuses on the case where the metric space is known, finite of size $n$ and the adversary is oblivious. In this case, they give a randomized $\mathcal{O}(\log^2 n + \log \Delta)$-competitive algorithm, where $\Delta$ is the aspect ratio of the metric space. This was later improved to $\mathcal{O}(\log n)$ by \cite{SODA:AzaChiKap17}. These algorithms use the same approach of embedding the metric space in a tree with logarithmic distortion. Regarding lower bounds, \cite{SODA:AzaChiKap17} first gave an $\Omega(\sqrt{\log n})$ lower bound which was later improved to $\Omega(\log n /  \log \log n)$ by Ashlagi et al \cite{ARCO:AACCGKMWW17}. These lower bounds hold for a randomized algorithm against an oblivious adversary. Moreover, because the latter bound uses $m = \Theta(n)$ requests, it also implies the $\Omega(\log m /  \log \log m)$ bound described above. When the delay cost is a concave function instead of a linear one, Azar et al \cite{SODA:AzaRenVai21} give a randomized $\mathcal{O}(\log n)$-competitive algorithm. As a more general approach, Deryckere and Umboh \cite{ARCO:DU23} consider the \textsf{MPMD-Set} problem, where the time cost at a given instant is a function of the unmatched requests.

To work around the logarithmic lower bound, many papers focus on special cases of online matching with delay. For example, for two-point metrics, Emek et al \cite{AC:ESW17} give an optimal deterministic $3$-competitive algorithm. When the metric space is a tree with height $h$, Azar et al \cite{SODA:AzaChiKap17} provide a deterministic $\mathcal{O}(h)$-competitive algorithm. If the requests' arrival times follow a Poisson distribution, Mari et al \cite{TCS:MPRS25} give an algorithm with constant competitiveness. For a $k$-point uniform metric space, Liu et al \cite{ISAAC:LPWW18} give a $\mathcal{O}(k)$-competitive algorithm which supports convex time cost functions.

A closely related problem to the \mpmd problem is the online Minimum-cost Bipartite Matching with Delay problem (\textsf{MBPMD}). In \textsf{MBPMD}, requests additionally carry a sign (positive or negative) and only requests of opposite sign can be matched. For this problem, Azar et al \cite{SODA:AzaChiKap17} first gave a randomized $\mathcal{O}(\log n)$-competitive algorithm along with a $\Omega(\log^{1/3} n)$ lower bound. This lower bound was later improved to $\Omega(\sqrt{\log n / \log \log n})$ by Ashlagi et al \cite{ARCO:AACCGKMWW17}. Many results for \mpmd, such as the $\mathcal{O}(m)$-competitive algorithm by Bienkowski et al \cite{AOA:BKLS18} and the algorithm with $\mathcal{O}(m^{0.59})$-competitiveness by Azar and Jacob-Fanani \cite{TCS:YJ20} have a variant with the same competitiveness for  \textsf{MBPMD}. More recently, Kuo \cite{WAOA:K24} gave an algorithm with $\tilde{\mathcal{O}}(\sqrt{m})$-competitiveness for \textsf{MBPMD} on a line.

\subsection{Algorithm overview}

We now give a brief overview of the techniques and ideas used by our algorithm. In essence, the algorithm partitions the requests which arrived so far into components. Requests inside each component are matched using a greedy strategy.

\textbf{Offline approach}: Our algorithm can be described as a way to adapt Wattenhofer and Wattenhofer's algorithm for offline perfect weighted matching \cite{ENDM:WW04} to the online setting. Their algorithm is similar to Boruvska's algorithm in that we try to cover vertices with components and do so within at most $\log n$ rounds. In each round, some components are merged together to create bigger components. The main difference is that Boruvska's algorithm finds a minimum spanning tree and ends with a single component remaining. Meanwhile the algorithm by \cite{ENDM:WW04} finds a stable matching and terminates when the size of all components remaining is even. Within a component, they show that requests can be paired efficiently following the order of an Euler tour.

\textbf{Online approach}: For the online setting, we adapt the algorithm by \cite{ENDM:WW04} with an overall approach similar to how Gallager et al adapted the Boruvka algorithm to run in the distributed setting with the GHS algorithm \cite{TPLS:GHS83}. More specifically, we consider each component as operating independently. When a component of odd size encounters another component it can merge with, it waits for a period of time proportional to the distance between them before merging. We give each component a rank, similar to the disjoint-set data structure, so that merging only happens in the direction of increasing rank.

\textbf{Matching inside a component}: In the offline setting, finding a perfect matching inside a component can be done efficiently using an Euler tour. This is not the case in the online setting, as some requests may join the component late. To work around this issue, we pair requests inside a component using a greedy matching algorithm. Although this greedy approach is not efficient in the general setting \cite{SJC:RT81}, we show that the matching obtained is a good approximation in this context.

\textbf{Minimizing waiting trees}: If a component $C_A$ wants to merge with a component $C_B$, but $C_B$ has a lower rank than $C_A$, then $C_A$ must wait for $C_B$'s rank to increase (or become irrelevant). This relation of $C_A$ waiting for $C_B$ can be seen as an oriented \textit{waiting} edge from $C_A$ to $C_B$ in a graph. When looking at the set of all these \textit{waiting} edges, they form a forest. By a procedure called \textit{Waiting tree pruning}, we make sure that all trees in this forest have a logarithmic size at all times. This ensures that the overall waiting time can be amortized.

\textbf{Regular and special edges}: A component always tries to merge with its closest \textit{compatible} component. When such a merge occurs, the set of edges added are called \textit{regular edges}. The idea, based on \cref{lemma:decomp-shortest-opt} and formally proven by \cref{lemma:bound-regular} is that the total weight of regular edges can be bounded by the cost of the optimal matching. However, this approach has many shortcomings when used in the online setting: a component may incorrectly identify its closest compatible component because some requests arrived late. To address this, we introduce \textit{special edges} and merge problematic components using them. The idea, formally proved by \cref{lemma:bound-special}, is that the weight of special edges is small relative to the weight of the components they are merging. 
\section{Preliminaries}

\subsection{Metric Space}

We consider a metric space $\mathcal{M} = (S, g)$, where $g$ is a distance function on $S$. In the \mpmd problem, an input $\mathcal{I} = (r_i)_{i \in [\![1;m]\!]}$ arrives in an online way, where a request $r_i$ arrives at time $t(r_i) \geq 0$ and at position $x(r_i) \in S$. The algorithm does not know a priori $m$, which is guaranteed to be even, or $S$. It can only query the distance between the position of requests that have already arrived. We note that we use $g$ instead of $d$ for the distance, as $g$ will be rarely used while we reserve $d$ for the time-augmented distance defined later.

The algorithm must produce a perfect matching incrementally. For any two requests $p$ and $q$, they can be matched together at any time $t \geq \max(t(p), t(q))$. Given a set of requests $(p_i, q_i, t_i)_{i \in [\![1;m/2]\!]}$ matched by an algorithm, the cost of the matching is:
\begin{align*}
    \sum_{i=1}^{m/2} \left( g(x(p_i), x(q_i)) + |t_i - t(p_i)| + |t_i - t(q_i)| \right)
\end{align*}

We define the competitiveness of an algorithm as the worst-case ratio of the cost of this algorithm compared to an optimal algorithm that knows the entire input ahead of time.

\textbf{Time-augmented metric space}:

We can consider adding the time as part of the metric space: Let us define $S' = S \times \mathbb{R}_+$ and $\mathcal{M'} = (S', d)$ where:
\begin{align*}
    d((p, t_1),(q, t_2)) = g(p,q) + |t_1 - t_2|
\end{align*}
One can check that $\mathcal{M'}$ also defines a metric space. Moreover, the optimal solution for regular perfect matching on $\mathcal{M'}$ is closely tied to the optimal online solution for $\mathcal{M}$ in the following way, which is formally proven in the appendix:

\begin{restatable}{lemma}{OnlineOfflineSame}\label{lemma:online-offline-same}
Consider an input $\mathcal{I} = (r_i)_{i \in [\![1;m]\!]}$, let $OPT$ be the cost of the optimal solution for \mpmd on $\mathcal{I}$ in $\mathcal{M}$ and $OPT'$ be the cost of the optimal solution for the offline min-cost perfect matching on $\mathcal{I}$ in $\mathcal{M'}$. Then $OPT = OPT'$.
\end{restatable}

For an edge $e = (u,v)$, with $u$ and $v$ in $S'$, we call $w(e)$ the weight of the edge and define it as $w(e) = d(u,v)$. For a set of edges $E$, we define $w(E)$ as $\sum_{e \in E} w(e)$.

We consider $V_{req}$ the multiset of all time-augmented requests an algorithm receives. We have $|V_{req}| = m$ and remark that the input to a \mpmd instance is uniquely defined by $V_{req}$.

We now give a useful lemma regarding the time-augmented space:
\begin{lemma} \label{lemma:arrival-time}
    Let $u \in S'$ be a request arriving at time $t(u)$. Then for any $l > 0$, by time $t(u) + l$, all requests $v$ such that $d(u,v) \leq l$ have arrived.
\end{lemma}

\begin{proof}
    Let $l > 0$ and $v \in S'$ be a request such that $d(u,v) \leq l$. We have $|t(v) - t(u)| \leq d(u,v) \leq l$. This implies that the arrival time $t(v)$ satisfies $t(v) \leq t(u) + d(u,v) \leq t(u) + l$.
\end{proof}

\subsection{Component Decomposition}

\begin{figure}[h]
\centering
\includegraphics[scale=0.6]{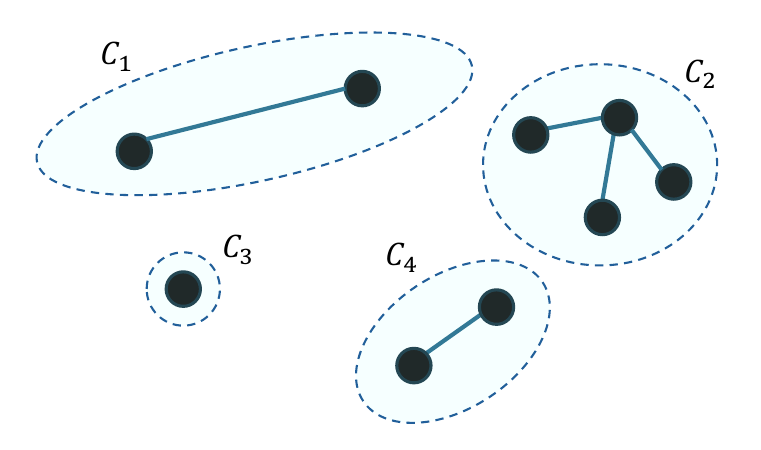}
\caption{Component decomposition of a graph with $9$ vertices into $4$ components}
\label{figure:component-decomposition}
\end{figure}

Our work relies extensively on component decomposition, i.e grouping multiple requests together inside a single component, and merging components together. This approach was already used successfully for offline weighted matching \cite{STOC:SupPlaRei80, ENDM:WW04}. Given a set of vertices $V$, we partition it into components, such that a component $C$ is defined by:
\begin{itemize}
    \item $V(C)$: the set of vertices it covers
    \item $T(C)$: a spanning tree over $V(C)$
\end{itemize}
When there is no ambiguity, we may sometimes identify a component with its set of vertices or its spanning tree. We also define $w(C) = w(T(C))$, the weight of its spanning tree. For the sake of our algorithm, when vertices are in a time-augmented space, we also define the following properties and give an informal description of their use:
\begin{itemize}
    \item $\rank(C)$: the rank of the component, which is an indicator of the size of $V(C)$. We want to make sure that $\rank(C) \leq \log |V(C)|$.
    \item $\nrank(C)$: nearby rank of the component. This can be set to $\bot$ or to a positive integer. When not $\bot$, it indicates that there is a "close by" component to $C$ whose rank is at least $\nrank(C)$. For simplicity, $\forall p \geq 1$, we define $\max(\bot, p)$ as $p$.
    \item $repr(C)$: representative of $C$. This is a vertex in $V(C)$ which is used to help identify $C$.
    \item $t_{max}(C) = \max_{v \in V(C)} t(v)$: the latest time of arrival of a request in $C$.
\end{itemize}
By partitioning $V$ into multiple components, we then define a component decomposition $\comp$ over $V$ as a set of components such that $V = \mathbin{\mathaccent\cdot\cup}_{C \in \comp} V(C)$. We also define the weight of the decomposition $w(\comp) = \sum_{C \in \comp} w(C)$.

When a component has an even number of vertices ($|V(\mathcal{C})|$ is even), we call it an even component. Otherwise we call it an odd component. We naturally extend the distance $d$ to components (we note that $d$ does not satisfy the triangle inequality on components):
\begin{align*}
    \forall C_1, C_2 \in \comp, \ d(C_1, C_2) = \min_{v_1 \in C_1, v_2 \in C_2} d(v_1, v_2)
\end{align*}

We now consider a new distance $D$ on components, which we call the \textit{compressed} distance. $D$ is obtained by compressing even components to a single point. To be more precise, for any two components $C_A, C_B$, and $l \geq 0$, $D(C_A, C_B) = l$ if and only if there exists $k \geq 1$, and components $C_A = C_1, C_2, \ldots, C_k = C_B$ such that $C_2, \ldots, C_{k-1}$ are even, $l = \sum_{i=1}^{k-1} d(C_i, C_{i+1})$ and no such path of lower weight exists. We note that $D$ is technically not a metric distance (it does not always satisfy the triangle inequality).

One of the main lemmas used for the analysis of our algorithm is the following. We note that it is a rephrasing of Lemma 2 from \cite{ENDM:WW04} and is formally proven in the appendix:
\begin{restatable}{lemma}{DecompShortestOpt}\label{lemma:decomp-shortest-opt}
    Let $\comp$ be a component decomposition over $V$. Let $\comp_{odd} \subseteq \comp$ be the set of odd components in $\comp$. Let $W_{OPT}$ be the weight of the optimal min-cost perfect matching over $V$. Moreover, for all $C \in \comp_{odd}$, we define:
    \begin{align*}
        r_C = \min_{C' \in \comp_{odd} \setminus \{C\}} D(C, C')
    \end{align*}
    Then:
    \begin{align*}
        \sum_{C \in \comp_{odd}} r_C \leq 2 W_{OPT}
    \end{align*}
\end{restatable}

Finally, we give the following somewhat simple but really useful lemma. This lemma follows from the fact that $D$ is the distance obtained by $d$ from compressing some components to a single point, therefore $D$ never exceeds $d$:
\begin{lemma}\label{lemma:comp-distance}
    Let $C_1$ and $C_2$ be two components in a component decomposition. For any $u \in C_1, v\in C_2$:
    \begin{align*}
        D(C_1, C_2) \leq d(u,v)
    \end{align*}
\end{lemma}
\section{Warmup}

In this section, we will briefly describe the algorithms and ideas our online procedure is based on as they appear in the offline setting. We note that the main difficulty in this paper lies in adapting these algorithms to the online setting and not their offline version.

\subsection{Greedy matching}

We consider the following greedy algorithm for the offline matching:
\begin{algobox}{\textsc{GREEDY}}
    \algoHead{Greedy perfect matching for a metric graph $G$}
    \begin{algorithmic}[1]
        \While{there are unmatched vertices}
            \State Let $u$ and $v$ be the closest unmatched vertices in $G$
            \State Match $u$ and $v$
        \EndWhile
    \end{algorithmic}
\end{algobox}

Although this algorithm is simple, it is also inefficient. In fact, Reingold and Tarjan \cite{SJC:RT81} showed that there exist metric graphs where \textsc{GREEDY} returns a matching of weight $\Theta(n^{\log 1.5}) \cdot OPT$ where $n$ is the number of vertices and $OPT$ the weight of the optimal matching, this bound is known to be tight. In their paper, Azar and Jacob-Fanani \cite{TCS:YJ20} use a slightly modified version of \textsc{GREEDY} to design an algorithm for matching with delay with competitiveness $\mathcal{O}(\varepsilon^{-1} \cdot n^{\log (1.5 + \varepsilon)}) \approx \mathcal{O}(n^{0.59})$ . Our goal is to obtain a poly-logarithmic competitive algorithm, hence using an approach based on \textsc{GREEDY} seems like a bad idea. However, we show that the output of \textsc{GREEDY} can be bounded in an efficient way by the weight of the optimal traveling salesman tour:

\begin{theorem}\label{theorem:greedy-offline-bound}
    Consider a metric graph $G$. Let $M$ be a matching returned by \textsc{GREEDY} and $w(M)$ its weight. Let $OPT_{TSP}$ be the optimal weight of the traveling salesman tour for graph $G$, then:
    \begin{align*}
        w(M) \leq \frac{1}{4}(\lceil \log n \rceil + 1) \cdot OPT_{TSP}
    \end{align*}
\end{theorem}

\begin{proof}
    Our proof relies on a lemma used to analyze TSP heuristics by Rosenkrantz et al \cite{SJC:RSL77}:

    \begin{lemma}[Lemma (1) of \cite{SJC:RSL77}] \label{lemma:metric-tsp}
         Suppose that for a metric graph $G = (V,g)$ with $n$ nodes, there is a mapping assigning each node $p$ a number $l_p$ such that the following two conditions hold:
         \begin{itemize}
             \item $g(p,q) \geq \min(l_p, l_q) \ \forall p,q \in V \ \text{with} \ p \ne q$
             \item $l_p \leq \frac{1}{2} OPT_{TSP} \ \forall p \in V$
         \end{itemize}
         Then $\sum l_p \leq  \frac{1}{2}(\lceil \log n \rceil + 1) OPT_{TSP}$
    \end{lemma}
    Let $p \in V$, we consider $q \in V$ the node matched with $p$ in $M$. We set $l_p = l_q = g(p,q)$ the length of the matched edge. Let us show that this definition satisfies the requirements of \cref{lemma:metric-tsp}:
    \begin{itemize}
        \item Let $p,q$ be two distinct nodes in $G$. Without loss of generality, let us assume that \textsc{GREEDY} matched $p$ before $q$ (if they got matched together any order is fine). Because right before $p$ was matched, both $p$ and $q$ were unmatched and the greedy takes the unmatched edge of smallest weight, it means that the edge used to match $p$ has weight at most $g(p,q)$. Therefore $g(p,q) \geq l_p \geq \min(l_p, l_q)$.
        \item Let $p \in V$, we assume $p$ got matched to some $q \in V$, so $l_p = g(p,q)$. The TSP consists of two paths between $p$ and $q$. Given that we are in a metric graph, we get $OPT_{TSP} \geq 2 g(p,q)$ and therefore $l_p \leq \frac{1}{2} OPT_{TSP}$
    \end{itemize}
    We can thus apply \cref{lemma:metric-tsp} on $(l_p)$:
    \begin{align*}
        \sum l_p \leq  \frac{1}{2}(\lceil \log n \rceil + 1) OPT_{TSP}
    \end{align*}
    We remark that by the definition of $(l_p)$, the weight of every matched edge appears exactly twice (one for each endpoint) in $(l_p)$, therefore $\sum l_p = 2 w(M)$. The theorem follows from the last two equations.
\end{proof}

The fact that we can bound the cost of the greedy matching in some way by the length of the optimal TSP tour is a key idea for our overall algorithm. Because we are in a metric graph, the length of this TSP tour can be bounded by some spanning tree over the vertices of $G$. After adapting this algorithm to the online setting, the main challenge that remains is to make sure we only run the greedy on subgraphs such that the sum of spanning trees covering them can be bounded in some way by the optimal online matching cost.

\subsection{Component based matching}

The previous bound on the greedy algorithm relies on the fact that we can give a proper upper bound on the weight of spanning trees covering all requests. We remark that the approximation algorithm for minimum weighted matching given by Wattenhofer and Wattenhofer \cite{ENDM:WW04} has such property. Given a forest $\mathcal{F}$, the component decomposition based on it is the one where each tree in the forest is its own component. We give a simplified version of their protocol. By "closest odd component", we mean closest using the distance $D$ which compresses even components to a single point. We note that removing edges until there are no cycles is only relevant if there are ties when considering $P$.

\begin{algobox}{\textsc{OFFLINE\_COMPONENT\_MATCHING}}
    \begin{algorithmic}[1]
        \State $\mathcal{F}_0 \gets \emptyset$
        \For{$i = 0...\lfloor \log |V| \rfloor$}
            \State Let $\comp_i$ be the component decomposition based on $\mathcal{F}_i$
            \For{each component $C \in \comp_i$}
                \State Match vertices in $V(C)$ until there is at most one unmatched vertex left
            \EndFor
            \State $\mathcal{F}_{i+1} \gets \mathcal{F}_{i}$
            \For{each odd component $C \in \comp_i$}
                \State Let $P$ be a shortest path from $C$ to its closest distinct odd component
                \State $\mathcal{F}_{i+1} \gets \mathcal{F}_{i+1} \cup P$
            \EndFor
            \State Remove edges in $\mathcal{F}_{i+1}$ until it contains no more cycles
        \EndFor
    \end{algorithmic}
\end{algobox}

This algorithm keeps merging odd components between each other until there are only even components left. The main idea for the analysis is that using \cref{lemma:decomp-shortest-opt}, one can prove that the weight of the edges added to the forest in each iteration is at most $W_{OPT}$. The authors of \cite{ENDM:WW04} match vertices using an Euler tour technique. However, because we can bound the weight of a spanning tree on the vertices, we can instead use our greedy matching algorithm. One can prove that with this approach, we get a $\mathcal{O}(\log^3 |V|)$-approximation for the min-cost perfect weighted matching problem.

The advantage of using the greedy matching instead of an Euler tour, and the approximation algorithm from \cite{ENDM:WW04} instead of the blossom algorithm, is that these algorithms use simple techniques. To be more precise, both of them only require knowledge of the closest vertex/component of a given vertex/component to produce a matching. This is something that can potentially be achieved with an online algorithm without too much overhead, compared to approaches based on alternating paths, which are difficult to adapt online. 

The rest of this paper is therefore dedicated to explaining how to adapt these two algorithms to the online setting to obtain a $\mathcal{O}(\log^5 m)$-competitive algorithm.

\section{Greedy algorithm for online matching}

Similarly to what Azar and Jacob-Fanini did \cite{TCS:YJ20}, a natural idea to adapt this offline greedy algorithm to the online setting would be to wait for $d(u,v)$ units of time after an edge $e=(u,v)$ appears before matching its endpoints. However, this approach fails because we are not guaranteed to match the closest pair of unmatched vertices. Instead we show that by slightly modifying the greedy algorithm and waiting twice that duration: $2d(u,v)$, we get the desired property.

\begin{algobox}{\textsc{GREEDY\_ITERATION}(t)}
    \algoHead{Iteration of an online perfect matching algorithm at time $t$}
    \begin{algorithmic}[1]
        \If{$G$ has at least $2$ unmatched requests}
            \For{each unmatched request $u \in S'$}
                \State Let $v$ be the closest unmatched request to $u$ according to $d$
                \If{$t \geq t(u) + \textbf{2} \cdot d(u,v)$}
                    \State Match $u$ and $v$
                \EndIf
            \EndFor
        \EndIf
    \end{algorithmic}
\end{algobox}

\begin{algobox}{GREEDY\_ONLINE}
    \algoHead{Online perfect matching for a metric space}
    \begin{algorithmic}[1]
        \State \textbf{At every moment t}
        \State Run \textsc{GREEDY\_ITERATION}(t)
    \end{algorithmic}
\end{algobox}

We show that this online greedy algorithm achieves properties similar to the offline version (\cref{theorem:greedy-offline-bound}), but with a worse constant factor. The formal proof uses the same idea as the offline version and is given in the appendix.

\begin{restatable}{lemma}{GreedyUpperBound}\label{theorem:greedy-upper-bound}
    Let $C_{GREEDY}$ be the total cost of the matching achieved by \textsc{GREEDY\_ONLINE} and let $OPT_{TSP}$ be the weight of the optimal TSP tour for $d$ over all requests matched by the algorithm. Then:
    \begin{align*}
        C_{GREEDY} \leq \frac{5}{2}(\lceil \log n \rceil + 1) OPT_{TSP}
    \end{align*}
\end{restatable}

We described an online matching algorithm whose cost is bounded, up to a logarithmic factor, to the weight of a tour on its input. The rest of this paper focuses on building on this foundation to get a $\mathcal{O}(\log^5 m)$-competitive algorithm.
\section{Algorithm}

We now give a proper description, along with pseudocode of our main online matching algorithm. We introduce many notations which, while not useful for the algorithm itself, are necessary for its analysis. As a brief description of our algorithm, it manages and updates a component decomposition $\comp$ over the set of requests which arrived so far. Requests are then merged within components using the greedy algorithm described previously.

\subsection{Overall structure}

We now describe the overall architecture of our algorithm. For simplicity, our algorithm is described as running continuously (i.e at all times). But at any given time, excluding new requests coming, only a finite amount of known events are planned to happen. Therefore, this algorithm can be adapted to run only at some finite set of points in time.

\begin{algobox}{\textsc{ONLINE\_MATCHING}}
    \algoHead{Deterministic algorithm for online matching with delay}
    \begin{algorithmic}[1]
        \State $\comp \gets \emptyset$, $\forall r \geq 0:  S_r \gets \emptyset, R_r \gets \emptyset$
        \State \textbf{At every moment $t$}
        \State \textsc{RECEIVE\_NEW\_REQUESTS}(t)
        \State \textbf{repeat}
        \State \hspace{\algorithmicindent} \textsc{COMBINE\_COMPONENTS}(t)
        \State \hspace{\algorithmicindent} \textsc{PRUNE\_WAITING\_TREES}(t)
        \State \textbf{until} no merge happened
        \State \textsc{RUN\_GREEDY}(t)
    \end{algorithmic}
\end{algobox}

We will provide a summary of each function used in this algorithm:
\begin{itemize}
\item \textsc{RECEIVE\_NEW\_REQUESTS}: the algorithm looks for requests that have just arrived at time $t$. For each of these requests, it creates a new component of rank $0$ containing this request as a singleton.
\item \textsc{COMBINE\_COMPONENTS}: also called the combining step, the algorithm looks at each odd component, then it identifies the closest compatible component to this component. If they are close enough related to the current time and the rank of the target component is at least as big as the current component's rank, then they merge together.
\item \textsc{PRUNE\_WAITING\_TREES}: also called the pruning step, if the closest component to an odd component is another odd component with lower rank, the former component must wait for the latter's rank to be big enough before merging. The pruning step ensures that for each active (non-waiting) component, at most a logarithmic amount of odd components are waiting on it.
\item \textsc{RUN\_GREEDY}: As soon as a component has two unmatched requests, we run the greedy algorithm on these requests. Due to the previous component decomposition, we are guaranteed that we can bound the space and time cost when running the greedy algorithm in this setting.
\end{itemize}

\subsection{Handling new requests}

As explained above, when receiving a new request, we create a simple component of rank $0$ containing it, then add it to the set of components.

\begin{algobox}{RECEIVE\_NEW\_REQUESTS(t)}
    \algoHead{Handle arrival of new requests}
    \begin{algorithmic}[1]
        \For{each new request $v$}
            \State Let $C$ be a new component such that $V(C) = \{(v, t)\}$, $\rank(C) = 0$, $\nrank(C) = \bot$ and $repr(C) = (v,t)$
            \State $\comp \gets \comp \cup \{C\}$
        \EndFor
    \end{algorithmic}
\end{algobox}

\subsection{Merging components}

\begin{figure}[h]
\centering
\includegraphics[scale=0.6]{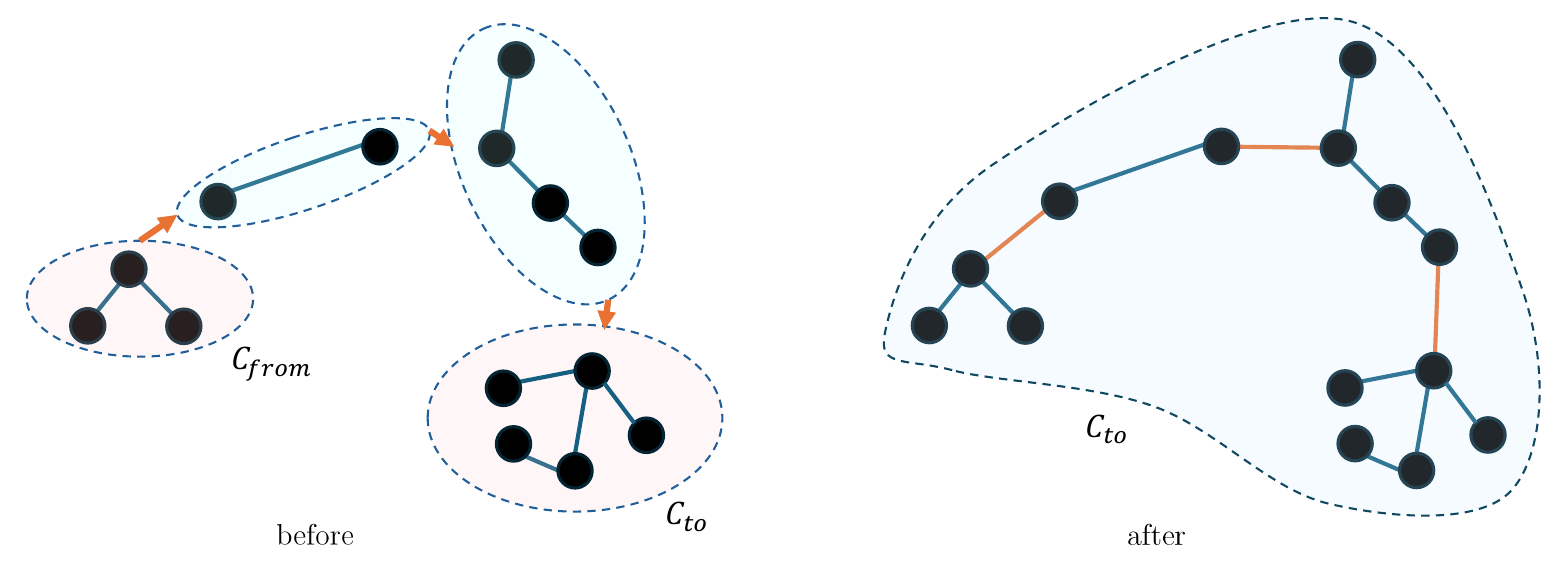}
\caption{Merging $C_{from}$ into $C_{to}$. The newly added edges (in orange) are given the rank $r$ of the merge}
\label{figure:merging}
\end{figure}

Merging components is a subprocedure of the combining and pruning steps. This procedure is always invoked when a component's set of vertices or spanning tree is updated. It takes as input a source $C_{from}$, a destination $C_{to}$ and a rank $r$. We say that $C_{from}$ is merged into $C_{to}$. Moreover, there might also be even components along the path $P$ from $C_{from}$ to $C_{to}$ which are merged into $C_{to}$. $P$ is chosen as the shortest path from $C_{from}$ to $C_{to}$ according to the distance $D$.

Edges newly added to the component decomposition by a merge of rank $r$ are called edges of rank $r$. The merging procedure has two variants whose pseudocode is described below:
\begin{itemize}
    \item A special merge: Edges added by a special merge of rank $r$ are added to the set $S_r$ and are called special edges. We use special merges in a way that the weight of special edges is minimal compared to the weight of the components being merged (see \cref{lemma:bound-special}).
    \item A regular merge: Edges added by a regular merge of rank $r$ are added to the set $R_r$ and are called regular edges. We use regular merges in a way that the weight of the edges can be bounded by the optimal matching cost (see \cref{lemma:bound-regular}). A regular merge performs an additional modification compared to a special one: given $l = D(C_{from}, C_{to})$, it makes sure all components within distance $l / (r + 1)$ have their nearby rank set to at least $r$.
\end{itemize}

\begin{algobox}{SPECIAL\_MERGE($C_{from}, C_{to}, r$)}
    \algoHead{Merge component $C_{from}$ into $C_{to}$ with special edges}
    \begin{algorithmic}[1]
        \State Let $P$ be the shortest path from $C_{from}$ to $C_{to}$
        \State Let $D$ be the set of even components along $P$ (excluding endpoints)
        \State Let $E$ be the set of edges in $P$ outside of components
        \State $V(C_{to}) \gets V(C_{to}) \cup V(C_{from}) \cup \bigcup_{C \in D} V(C)$
        \State $T(C_{to}) \gets T(C_{to}) \cup T(C_{from}) \cup E \cup \bigcup_{C \in D} T(C)$
        \State $S_{r} \gets S_{r} \cup E$
        \State $\comp \gets \comp \setminus (D \cup \{C_{from}\})$
    \end{algorithmic}
\end{algobox}

\begin{algobox}{REGULAR\_MERGE($C_{from}, C_{to}, r$)}
    \algoHead{Merge component $C_{from}$ into $C_{to}$ with regular edges}
    \begin{algorithmic}[1]
        \State Let $P$ be the shortest path from $C_{from}$ to $C_{to}$
        \State Let $D$ be the set of even components along $P$ (excluding endpoints)
        \State Let $E$ be the set of edges in $P$ outside of components
        \State $H \gets \{C \in \comp, C \ne C_{from}, D(C_{from}, C) < D(C_{from}, C_{to}) / (r + 1)\}$
        \For{$C \in H \setminus D$}
            \State $\nrank(C) \gets \max (\nrank(C), r)$
        \EndFor
        \State $V(C_{to}) \gets V(C_{to}) \cup V(C_{from}) \cup \bigcup_{C \in D} V(C)$
        \State $T(C_{to}) \gets T(C_{to}) \cup T(C_{from}) \cup E \cup \bigcup_{C \in D} T(C)$
        \State $R_{r} \gets R_{r} \cup E$
        \State $\comp \gets \comp \setminus (D \cup \{C_{from}\})$
    \end{algorithmic}
\end{algobox}

We can now describe the combining step. A crucial point for the combining step is the notion of compatible components. This concept replaces the notion of odd components used by the offline algorithm when looking to merge components with each other. Let $C_1$ be an odd component. We say that another component $C_2 \ne C_1$ is \textbf{compatible} with $C_1$ if at least one of the following three conditions is satisfied:
\begin{itemize}
    \item $C_2$ is odd: we can always try to merge with an odd component
    \item $\rank(C_2) \geq \rank(C_1)$: $C_2$ used to be odd previously but we missed it because $C_1$ was "late", this can be seen as a way to catch up on it.
    \item $\nrank(C_2) > \rank(C_1)$: we know there is a component with rank at least $\nrank(C_2)$ next to $C_2$, we can merge with $C_2$ then merge with this nearby component at a low cost.
\end{itemize}
We note that if $C_2$ is not compatible with $C_1$, it implies that it is even and $\rank(C_2) < \rank(C_1)$. 

During the combining step, we try to merge odd components to their closest compatible component. Let $C_1$ be an odd component and $C_2$ be its closest compatible component. To simplify the algorithm, if $C_1$ has no compatible component (it can happen if the number of requests which arrived so far is odd), we define $C_2$ as being a fake component with distance $+ \infty$ from $C_1$. Let $l = D(C_1, C_2)$, the first step is to wait until time $t_{max}(C_1) + 2l$ to make sure we did not miss any closer component which arrived late. The second step depends on the nature of $C_2$ and the surroundings of $C_1$:
\begin{itemize}
    \item If there exists a component $C_3 \ne C_1$ such that $C_3$ is close enough to $C_1$ but $t_{max}(C_3)$ is really high, we cannot ignore $C_3$ as it might cause us to miss out on better compatible components which did not arrive yet. So we use a special merge to merge $C_3$ with $C_1$.
    \item If $\nrank(C_2) > \rank(C_1)$, we merge $C_1$ into $C_2$ with a regular merge. Then we call the nearby fixup procedure which keeps merging $C_2$ with its closest big enough component using special merges until the resulting component has nearby rank $\bot$.
    \item If $\rank(C_2) \geq \rank(C_1)$, we can immediately merge $C_1$ into $C_2$.
    \item Otherwise, if none of the cases above apply, it implies that $C_2$ is an odd component and $\rank(C_2) < \rank(C_1)$. We cannot merge $C_1$ into $C_2$ yet as it would break some invariants we rely on for the analysis of our algorithm. Instead, we say that there is a \textit{waiting} edge from $C_1$ to $C_2$, which will be considered during the pruning step and take no further action regarding $C_1$ in the combining step.
\end{itemize}

\begin{algobox}{COMBINE\_COMPONENTS(t)}
    \algoHead{Merge components that are close to each other }
    \begin{algorithmic}[1]
        \For{each odd component $C_1 \in \comp$}
            \State $C_2 \gets \arg \min \{ D(C_1, C), C \in \comp$, \text{$C$ compatible with $C_1$}\}
            \State $l \gets D(C_1, C_2)$
            \If{$t \geq t_{max}(C_1) + 2 \cdot l$}
                \State $H \gets \{ C \in \comp, C \ne C_1, D(C_1, C) < l / (\rank(C_1) + 2) \}$
                \If{$\exists \ C_3 \in H, t_{max}(C_3) \geq t_{max}(C_1) + l$}
                    \State \textsc{SPECIAL\_MERGE}($C_3$, $C_1$, $\rank(C_1)$)
                \ElsIf{$\nrank(C_2) > \rank(C_1)$}
                    \State \textsc{REGULAR\_MERGE}($C_1$, $C_2$, $\nrank(C_2)$)
                    \State \textsc{NEARBY\_FIXUP}($C_2$)
                \ElsIf{$\rank(C_2) \geq \rank(C_1)$}
                    \If{$\rank(C_2) = \rank(C_1)$}
                        \State $\rank(C_2) \gets \rank(C_2) + 1$
                    \EndIf
                    \State \textsc{REGULAR\_MERGE}($C_1$, $C_2$, $\rank(C_2)$)
                \Else
                    \State Do nothing ($C_1 \to C_2$ is a waiting edge)
                \EndIf
            \EndIf
        \EndFor
    \end{algorithmic}
\end{algobox}

\begin{algobox}{NEARBY\_FIXUP($C_1$)}
    \algoHead{Merge a component to its bigger nearby component while necessary}
    \begin{algorithmic}[1]
        \While{$\nrank(C_1) \ne \bot$}
            \State Let $C_2$ be the closest component to $C_1$ with rank at least $\nrank(C_1)$ or nearby rank at least $\nrank(C_1) + 1$
            \State \textsc{SPECIAL\_MERGE}($C_1$, $C_2$, $\max( \rank(C_2), \nrank(C_2))$)
            \State $C_1 \gets C_2$
        \EndWhile
    \end{algorithmic}
\end{algobox}

\subsection{Pruning waiting trees}

We now describe the pruning step. We remark that in the combining step, for an odd component $C_1$ and its closest compatible one $C_2$, there is one case where we can pass the threshold time but do nothing. In this case we consider the oriented edge $C_1 \to C_2$ as a \textit{waiting} edge. We remark in this case that $C_2$ is odd and $\rank(C_2) < \rank(C_1)$, so \textit{waiting} edges are only directed towards a component of strictly lower rank (i.e there are no cycles). 

We can consider $\mathcal{W}$ as the set of all waiting edges between components. Because there are no cycles and each component has out-degree at most $1$ (each component is waiting on at most one other component), $\mathcal{W}$ is a forest. The pruning step looks at each tree in this forest, If there are two components with the same rank in the same tree, they are merged together with a regular merge along the tree path and their rank is increased. This ensures that not too many components are waiting, which could be detrimental to the competitiveness of our algorithm. An example of tree pruning can be seen with \cref{figure:tree-pruning}.

\begin{algobox}{PRUNE\_WAITING\_TREES(t)}
    \algoHead{Merge nodes inside a waiting tree if possible}
    \begin{algorithmic}[1]
        \State Let $\mathcal{W}$ be the forest of waiting edges
        \For{each waiting tree $\mathcal{T}$ in $\mathcal{W}$}
            \If{there exists $C_1, C_2 \in \mathcal{T}$ such that $\rank(C_1) = \rank(C_2)$}
                \State $r \gets \rank(C_1)$
                \State let $C_3$ be the LCA of $C_1$ and $C_2$ in $\mathcal{T}$
                \State $\rank(C_3) \gets r + 1$ 
                \State $H \gets \{ C \in \mathcal{T}, \rank(C) \leq \rank(C_1), C_3 \text{\ is an ancestor of C} \}$
                \For{each $C \in H \setminus \{ C_3 \}$ by non-decreasing rank order}
                    \State \textsc{REGULAR\_MERGE}($C$, $C_3$, $r + 1$)
                \EndFor
            \EndIf
        \EndFor
    \end{algorithmic}
\end{algobox}

\begin{figure}[h]
\centering
\includegraphics[scale=0.6]{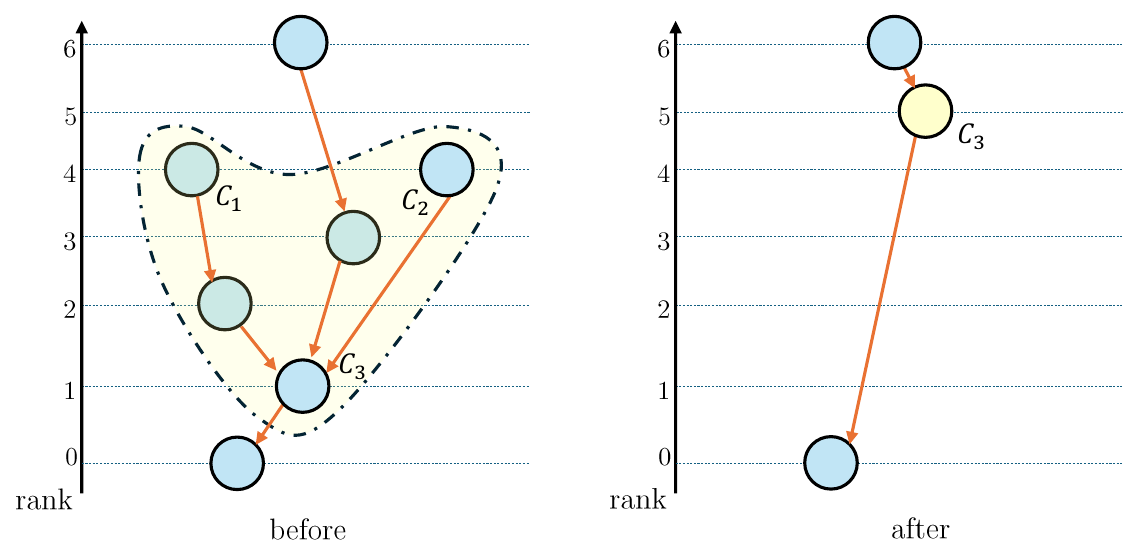}
\caption{Tree pruning: Components $C_1$ and $C_2$ have the same rank $4$, their least common ancestor is $C_3$. All components with ancestor $C_3$ and rank at most $4$ get merged together and the resulting component has rank $5$. The orange edges inside the new component get added to $T(C_3)$ with rank $5$.}
\label{figure:tree-pruning}
\end{figure}

\subsection{Greedy algorithm}

The last step for our online algorithm is to match requests together. This task is done by our greedy algorithm described above. We run multiple greedy instances, one for each component, and a greedy instance is uniquely identified by the representative $repr(C)$ of each component $C$. When a component $C$ has two or more unmatched requests not attached to any greedy instance, these requests are added two at a time to the greedy instance associated with $repr(C)$.

We remark that our greedy algorithm needs to assume that it can see requests as soon as they arrive. This is not the case here, so we modify the arrival time of the requests as seen by the greedy procedure to be the time they joined this instance instead of their real arrival time. Due to our design, each greedy instance always has an even number of requests.

\begin{algobox}{RUN\_GREEDY(t)}
    \algoHead{Update the greedy instances with new request then run them}
    \begin{algorithmic}[1]
        \For{each component $C \in \comp$}
            \While{$C$ has at least two unmatched requests not attached to any greedy instance}
                \State Let $r_1, r_2$ be two such requests
                \State Identify $r_1$ with $r_1' = (x(r_1), t)$
                \State Identify $r_2$ with $r_2' = (x(r_2), t)$
                \State Attach $r'_1$ and $r'_2$ to the greedy instance on $repr(C)$ 
            \EndWhile
        \EndFor
        \For{each greedy instance with at least two unmatched requests}
            \State Run \textsc{GREEDY\_ITERATION}(t) on the requests in the greedy instance
        \EndFor
    \end{algorithmic}
\end{algobox}
\section{Analysis of the algorithm}

We consider running the algorithm on a fixed input. Let $OPT$ be the optimal online matching with delay cost (which contains both the time and space cost) an algorithm knowing in advance this input can achieve. We note that because of \cref{lemma:online-offline-same}, $OPT$ is also the minimal weight of an offline perfect matching on the time-augmented input. If there is any ambiguity, for $t \geq 0$, we denote by $D_t$ the function $D$ applied on the algorithm's component decomposition at time $t$. In this section, we provide a detailed analysis of the algorithm, this analysis is divided into multiple parts:

\begin{itemize}
    \item \textit{Correctness}: We prove the correctness of the algorithm (it is well-defined and produces a perfect matching).
    \item \textit{Algorithm properties}: We describe and prove multiple properties and invariants satisfied by our algorithm. These are then built upon to provide the following results.
    \item \textit{Bounding edge weights}: We show that the size of the forest (spanning tree inside each component) generated by the algorithm can be bounded by $OPT$, up to a polylog factor.
    \item \textit{Bounding the algorithm's cost}: Using our previous bound on the forest weight, we can use it to get an upper bound on the time requests spend before joining a greedy invocation and finally bound the total cost of each greedy invocation.
\end{itemize}

\subsection{Correctness}

We first remark that our algorithm contains loops. We first show that these will never cause the algorithm to loop infinitely:
\begin{lemma}
    \textsc{ONLINE\_MATCHING} never gets stuck in a loop.
\end{lemma}

\begin{proof}
    We first remark that each component starts by containing a single request and each merge decreases the total number of components by at least one. Therefore, because there are in total $m$ requests, there can always be at most $m-1$ merges during the whole execution of the algorithm. So the main loop of \textsc{ONLINE\_MATCHING}, which performs an additional iteration every time a merge happens, will never loop infinitely. In a similar way, the loop of the \textsc{NEARBY\_FIXUP} procedure, which contains a special merge in its body, cannot iterate infinitely. These are the only two loops our algorithm has.
\end{proof}

We also remark that our algorithm makes assumptions about the existence of a component with specific properties in the \textsc{NEARBY\_FIXUP} procedure. The following lemma shows that these assumptions always hold, so our algorithm is well-defined. Because it requires additional observations about the nearby rank of a component, we defer its proof to later in the analysis.

\begin{restatable}{lemma}{AlgoWellDefined}
     The assumptions made in the \textsc{NEARBY\_FIXUP} procedure always hold.
\end{restatable}

We can now give the main result on the correctness, which shows that our algorithm solves the \mpmd problem:

\begin{lemma}
    \textsc{ONLINE\_MATCHING} produces a perfect matching.
\end{lemma}

\begin{proof}
    We first notice that assuming its input contains an even number of requests, the greedy algorithm produces a perfect matching within a finite amount of time. Let $t_{max}$ the maximal time of arrival of a request in the greedy input. Let $d_{max}$ the maximal distance between two requests in the greedy input. Then at time $t_{max} + 2d_{max}$, the algorithm would immediately match in pairs any requests not matched yet. Moreover, once two requests have been matched, they are not considered anymore. Therefore, by time $t_{max} + 2d_{max}$, the greedy algorithm computes a perfect matching.

    Looking at the overall algorithm, we note that it adds two requests to the greedy at a time, this ensures that the input of any greedy instance will only contain an even amount of requests. Therefore, all requests in an even component will be part of a greedy instance and odd components will have exactly one request not part of a greedy instance. It thus suffices to prove that some time after all requests arrive, there will be only even components remaining to conclude.

    As shown in the previous proof, the total number of merges the algorithm can perform is at most $m - 1$, so finite. Thus, we can find some time $t_{end}$ after all requests arrived such that the component structure does not change afterwards. We also consider $t_{last} \leq t_{end}$ the last arrival time of a request.

    Let us assume that after $t_{end}$, there is still at least one odd component. We note that the parity of $m$ must be the same as the number of odd components. Because $m$ is even, there must be an even number of odd components. Because there is at least one odd component, it means there are at least two such components. Let $C_1$ be the odd component with minimum rank and let $C_2$ be any other odd component. We note that a component can only wait on an odd component with strictly lower rank, therefore $C_1$ cannot wait on another component. Let $l = d(C_1, C_2)$, using \cref{lemma:comp-distance}, $l \geq D(C_1, C_2)$. Therefore after $t_{end}$ and by time $t_{last} + 2l$, we are guaranteed the algorithm will merge $C_1$ to its nearest compatible component, being $C_2$ or some other component. This is a contradiction because we assumed no merge happened after $t_{end}$. Therefore, there are no odd components remaining after $t_{end}$ and all even components will have a perfect matching done over their requests within a finite amount of time. This concludes the proof.
\end{proof}

\subsection{Algorithm Properties}

We follow with multiple observations that lay the foundation for our analysis of the algorithm.

\begin{claim}
    If a component has rank $r$, then it contains at least $2^r$ requests.
\end{claim}

\begin{proof}
    The proof is done by induction. When a request arrives, a component of rank $0$ is created for it, which satisfies the property. Subsequently, the rank of a component can only be increased in two ways. The first one is when two components of rank $r$ merge in the combining step, causing the resulting component to have rank $r+1$. The second one is in the pruning step, which happens when we have two components in one tree with the same rank $r$. These get merged together to get a component of rank $r+1$. In both cases, a component of rank $r+1$ is formed by a merge containing at least two components of rank $r$. Hence, the claim follows immediately by induction.
\end{proof}

\begin{corollary}\label{coro:max-comp-rank}
    The maximum rank a component can have is $\lfloor \log m \rfloor$.
\end{corollary}

We observe that previous work using a component-based approach in the offline setting \cite{STOC:SupPlaRei80, ENDM:WW04} could bound the maximum rank by $\log_3(3m/2)$. This is not the case here as we have to relax the conditions for merging to handle the online setting.

\begin{claim} \label{lemma:regular-invariants}
    A regular merge from $C_{from}$ into $C_{to}$ of rank $r$ is only called with $r > rank(C_{from})$ and $C_{to}$ being the closest compatible component to $C_{from}$.
\end{claim}

\begin{proof}
    We consider the three different cases where a regular merge happens:
    \begin{itemize}
        \item In the combining step, with $C_1$ and $C_2$. By definition $C_2$ is indeed the closest compatible component to $C_1$. In the first case we have $r = \nrank(C_2) > \rank(C_1)$. In the second case, we have $r = \rank(C_2) \geq \rank(C_1)$. Moreover, in the latter case, if $\rank(C_2)$ was equal to $\rank(C_1)$, it gets incremented, therefore $r > \rank(C_1)$.
        \item In the pruning step, along edges of the waiting tree. But by definition, for a waiting edge $C_A \rightarrow C_B$ to exist, $C_B$ must be the closest compatible component to $C_A$. Moreover, a regular merge is only called with components from $H$, which by definition have rank at most $\rank(C_1)$ while $r = \rank(C_1) + 1$ which therefore satisfy the property.
    \end{itemize}
\end{proof}

\begin{claim}\label{lemma:nearby-greater}
    Let $C$ be a component such that $\nrank(C) \neq \bot$, then $C$ is an even component and $\nrank(C) > \rank(C)$
\end{claim}

\begin{proof}
    We first consider the time a component is given a non-$\bot$ nearby rank, we can find components $C_{from}$, $C_{to}$ and a rank $r$ such that it happens in the regular merge from $C_{from}$ to $C_{to}$ at rank $r$. Using \cref{lemma:regular-invariants}, $C_{to}$ is the closest compatible component to $C_{from}$. If a component $C$ gets its nearby rank modified by this part of the algorithm, it implies that it is in $H$ and therefore that $D(C_{from}, C) < D(C_{from}, C_{to}) / (\rank(C_{from}) + 2) \leq D(C_{from}, C_{to})$. By minimality of $C_{to}$, it implies that $C$ is not compatible with $C_{from}$. Therefore, $C$ is an even component and $\rank(C) < \rank(C_{from})$. Using \cref{lemma:regular-invariants} again, we get that $\rank(C_{from}) < r$,  which proves that $\rank(C) < r = \nrank(C)$.

    We claim that if any component merges into $C$, then it immediately merges into a higher rank component. Indeed, because $C$ is even, no component can merge into it in the pruning step as waiting edges are only between odd components. Moreover, in the combining, given that $C$ is even, for it to be compatible with $C_1$, either its nearby rank is at least $\rank(C_1) + 1$ or its rank is at least $\rank(C_1)$ in which case its nearby rank is also at least $\rank(C_1) + 1$. Therefore, the nearby fixup procedure gets called and $C$ gets eventually merged into a component of rank at least its nearby rank. This proves that as soon as a component is given a non-$\bot$ nearby rank, it stays as an even component and its rank never increases.
\end{proof}

We can now get similar invariants for special merges compared to \cref{lemma:regular-invariants}
\begin{claim}\label{lemma:special-invariants}
    A special merge from $C_{from}$ into $C_{to}$ of rank $r$ is only called with $r > \rank(C_{from})$. Moreover, all even components (excluding endpoints) on the shortest path from $C_{from}$ to $C_{to}$ have rank strictly less than $r$.
\end{claim}

\begin{proof}
    Special merges are only used at two locations in our algorithm. The first one is within the combining step, to merge component $C_3$ into $C_1$. We remark that $C_3 \in H$, therefore $D(C_1, C_3) < D(C_1, C_2)$. Because $C_2$ is the closest compatible component to $C_1$, it means that $C_3$ is not compatible with $C_1$, which implies that $\rank(C_3) < \rank(C_1)$ and $\rank(C_1)$ is the rank used for the special merge, which proves the lemma in this case. Moreover, all even components on the shortest path from $C_1$ to $C_3$ are also not compatible with $C_1$, therefore their ranks are also strictly less than $\rank(C_1)$. The other location is in the nearby fixup procedure, where $C_1$ is merged into $C_2$ with rank $r = \max(\rank(C_2), \nrank(C_2))$. We remark that we chose $C_2$ such that its rank or nearby rank is at least $\nrank(C_1)$, i.e $r \geq \nrank(C_1)$. Using \cref{lemma:nearby-greater}, $r \geq \nrank(C_1) > \rank(C_1)$, which proves the first point. Regarding even components on the path from $C_1$ to $C_2$, the fixup procedure did not consider them, thus their rank must be strictly less than $\nrank(C_1) \leq r$, this satisfies the second point.
\end{proof}

\begin{claim}\label{lemma:merge-rank-bigger}
    If a regular or special merge from $C_{from}$ into $C_{to}$ of rank $r$ is called, with $D$ being the set of components merged on the shortest path from $C_{from}$ to $C_{to}$, then: $$r \geq \max(\rank(C_{from}), \rank(C_{to}), \max_{C \in D} \rank(C))$$
\end{claim}

\begin{proof}
    The property $r > \rank(C_{from})$ is a direct consequence of \cref{lemma:regular-invariants,lemma:special-invariants}. Moreover, \cref{lemma:regular-invariants} proves that components in $D$ are not compatible with $C_{from}$, so their rank is less than $\rank(C_{from})$. Meanwhile, \cref{lemma:special-invariants} directly proves that $r > \max_{C \in D} \rank(C)$ for special merges.  We now prove that $r \geq \rank(C_{to})$.

    For regular merges, in the combining step, the rank used is either $\rank(C_{to})$ or $\nrank(C_{to})$ which using \cref{lemma:nearby-greater} is at least $\rank(C_{to})$. In the pruning step, the rank used is $\rank(C_1) + 1$, which is the rank of $C_{to} = C_3$.
    For special merges, the rank used is either $\rank(C_{to})$ or $\max(\rank(C_{to}), \nrank(C_{to}))$, which satisfies the property.
\end{proof}

\begin{claim}\label{lemma:merge-bigger-component}
    If a regular or special merge from $C_{from}$ into $C_{to}$ of rank $r$ is called, then within the same iteration, $C_{to}$ will be part of a component of rank at least $r$.
\end{claim}

\begin{proof}
    We consider the different locations in the algorithm where these merges happen. For most of these cases, this is done with $r = \rank(C_{to})$ which implies this property. The only exception is in the combining step, where a regular merge is done with rank $\nrank(C_{to})$. However, right after, the fixup procedure keeps merging $C_{to}$ until it reaches a component of rank at least $\nrank(C_{to})$, which satisfies this property. 
\end{proof}

We can now prove the well-behavior of our algorithm:
\AlgoWellDefined*

\begin{proof}
    We first remark that \cref{lemma:merge-rank-bigger,lemma:merge-bigger-component} show that the rank of a component is only increasing: if at some time $t$ a component of rank $r$ exists, then at any time $t' \geq t$, we are guaranteed that a component of rank $r' \geq r$ will exist.

    We now consider the fixup procedure. We have $\nrank(C_1) \neq \bot$. $C_1$ was given its nearby rank by a previous regular merge of rank $\nrank(C_1)$. Using the previous observation and \cref{lemma:merge-bigger-component}, it means there exists a component $C'$ of rank at least $\nrank(C_1)$. Moreover, using \cref{lemma:nearby-greater}, we have $\rank(C_1) < \nrank(C_1)$, therefore $C_1 \ne C'$. Thus, there always exists at least one other component with rank at least $\nrank(C_1)$, so the assumptions made by the fixup procedure hold.
\end{proof}

\begin{claim}\label{lemma:rank-inside-comp}
    Let $C$ be a component of rank $r$. If any regular or special edge $e$ with an endpoint in $V(C)$ gets added, its rank $r'$ will be at least $r$. Moreover, if $r' > r$, then component $C$ will be merged into a component of rank at least $r'$ within the same algorithm iteration.
\end{claim}

\begin{proof}
    Because edges are only added inside a merge procedure, this is a direct consequence of \cref{lemma:merge-rank-bigger,lemma:merge-bigger-component}.
\end{proof}

\begin{corollary}\label{coro:compo-edge-rank}
    A component $C$ of rank $r$ only contains edges of rank at most $r$.
\end{corollary}

\begin{lemma}\label{lemma:request-in-odd}
    Consider a component $C$ at some time $t$, if $C$ is odd, then there is a request in $V(C)$ which has always been part of an odd component until now. If $C$ is even, then there is a request in $V(C)$ for which the component it was in was always odd as long as its rank was strictly less than $rank(C)$.
\end{lemma}

\begin{proof}
    We first claim that when doing a regular or special merge, one of the two endpoints is an odd component. Indeed, in the combining step, one of $C_{from}$ or $C_{to}$ is always $C_1$ which is an odd component. When doing the nearby fixup procedure, we keep merging $C_1$ which originally was odd to a component $C_2$ until the nearby rank of $C_2$ is $\bot$. Because of \cref{lemma:nearby-greater}, components with a non-$\bot$ nearby rank are always even, so we keep merging an odd component to an even component, which results in an odd component.
    In the pruning step, we observe that any component in the waiting tree is odd, so all merges follows this property.

    We now prove this lemma by induction on the time of the last merge operation. We consider a component $C$ at some time $t$ right after its last merge operation. If the component never had a merge, it means it has rank $0$, and is therefore a singleton. The request it contains immediately satisfies the property. Otherwise, assume this property is satisfied by all components until the previous merge. As stated above, one of the two components during this merge is an odd component $C'$, with rank at most $rank(C)$. Therefore, using the property on $C'$, we can find a request $v$ which always has been inside an odd component up to this point. Whether $C$ is odd or even, $v$ satisfies the desired property for $C$.
\end{proof}

\begin{lemma}\label{lemma:nearby-after-edges}
    Let $C$ be a component such that $\nrank(C) \ne \bot$ at some time $t$. Then all edges adjacent to $V(C)$ added afterwards have rank at least $\nrank(C)$.
\end{lemma}

\begin{proof}
    We remark that if $C$ gets merged into a component of rank at least $\nrank(C)$, then using \cref{lemma:rank-inside-comp}, all edges added after this time will have rank at least $\nrank(C)$. Therefore, we only need to consider edges added before $C$ gets merged into a component with such rank.

    We first consider the case where $C$ gets merged as an even component on the path between $C_{from}$ and $C_{to}$ for a regular merge. This implies that $C$ is not compatible with $C_{from}$, so $\nrank(C) \leq \rank(C_{from})$. Because of \cref{lemma:regular-invariants}, the edge rank is strictly more than $\rank(C_{from})$, which satisfies our requirement. If a special merge is done inside the combining step, then $C$ is not compatible with $C_{to}$ which using the same argument as before implies that the edges added have rank at least $\nrank(C)$. Because $C$ is even, it cannot be the first endpoint for a merge in the combining step or any endpoint for a merge in the pruning step because all endpoints are odd components. Because of \cref{lemma:nearby-greater}, we have $\nrank(C) > \rank(C)$, therefore a merge can only happen with $C$ as $C_{to}$, and either right before the nearby fixup procedure, for which the edge rank is at least $\nrank(C)$ or inside the nearby fixup, for which the edge rank is also at least $\nrank(C)$. This completes the proof.
\end{proof}

\subsection{Bounding edge weights}

For $r \in [\![0, \lfloor \log m\rfloor ]\!]$, we define the following:
\begin{align*}
    \mathcal{F}_{r} = \bigcup_{k=1}^r ( R_k \cup S_k)
\end{align*}
$\mathcal{F}_{r}$ is the set of all edges with rank at most $r$. We note that \cref{coro:max-comp-rank,coro:compo-edge-rank} implies that the maximum rank an edge can have is $\lfloor \log m\rfloor$. So we can also define $\mathcal{F} = \mathcal{F}_{\lfloor \log m\rfloor }$ to be the set of all edges added to components. We start with a first observation:

\begin{lemma}
    $\mathcal{F}$ is a forest.
\end{lemma}

\begin{proof}
    We note that all regular or special edges we add are part of the shortest path (according to $D$) between two distinct components. As a consequence, no cycles are introduced when adding regular or special edges.
\end{proof}

For $i \in [\![0, \lfloor \log m\rfloor ]\!]$, $\mathcal{F}_{i}$ is a subset of $\mathcal{F}$, so is also a forest. We consider the offline component decomposition $\dec_i$ where each component $C$ is a connected component in the forest $\mathcal{F}_{\leq i}$, with $V(C)$ being the vertices of a tree in this forest and $T(C)$ its edges. An example of such decomposition is given in \cref{figure:online-decomposition}.

\begin{figure}[h]
\centering
\includegraphics[scale=0.75]{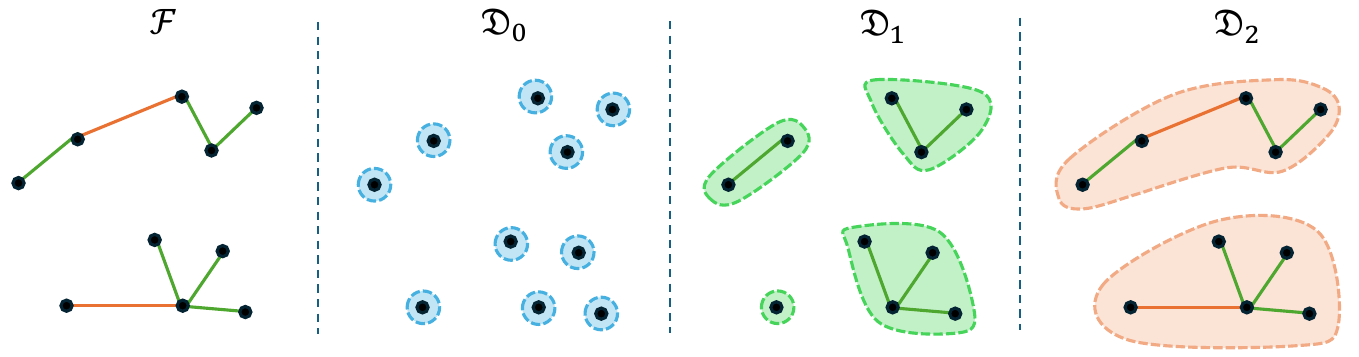}
\caption{Component decompositions for a given set of requests and edges. Green edges have rank $1$ while orange edges have rank $2$.}
\label{figure:online-decomposition}
\end{figure}

We claim that for any component in $\dec_i$, we can find a time in the execution of the online algorithm where this exact component existed. Moreover this time can be chosen right before it got merged into a higher-rank component or at the end of the execution of the algorithm:

\begin{lemma}\label{lemma:component}
    Let $i \in [\![0, \lfloor \log m\rfloor ]\!]$ and $C \in \dec_i$, then there is a time $t \geq 0$ during the execution of the algorithm where this exact component, with the same vertices and spanning tree, and rank at most $i$, existed during the execution of the algorithm. Moreover, we can choose this time $t$ to be either at the end of the execution of the online algorithm if the component is never used again, right before this component merges into a higher-rank component or right before its rank increases. 
\end{lemma}

\begin{proof}
    This is a direct consequence of \cref{lemma:rank-inside-comp,coro:compo-edge-rank}.
\end{proof}

We can now give an upper bound on the weight of regular edges:

\begin{lemma}\label{lemma:bound-regular}
    Let $i \in [\![1, \lfloor \log m\rfloor ]\!]$, then:
    \begin{align*}
        w(R_i) \leq 2(i+1) \cdot OPT
    \end{align*}
\end{lemma}

\begin{proof}
    An edge is part of $R_i$ if it was added by a regular merge from some component $C_{from}$ to $C_{to}$ with rank $i$. Let $A$ be the set of all components which were used as $C_{from}$ for a regular merge of rank $i$. We remark, using \cref{lemma:regular-invariants}, that all components in $A$ are odd and have rank strictly less than $i$. Moreover, after this merge, they are part of a component of rank at least $i$, therefore all components in $A$ are disjoint (by which we mean the set of vertices of each component are disjoint). Furthermore, using \cref{lemma:component}, every component in $A$ exactly matches one component in $\dec_{i-1}$, i.e $A \subseteq \dec_{i-1}$.

    For every component $C \in A$, we note $r_C$ the compressed distance from $C$ to the closest odd component in $\dec_{i-1}$. Moreover, we note $l_C$ the total weight of the edges of rank $i$ added from the regular merge with source $C$ and rank $i$. We note because of \cref{lemma:regular-invariants} that $l_C$ is the compressed distance from $C$ to the closest compatible component at the time of merging. We have:
    \begin{align*}
        w(R_i) = \sum_{C \in A} l_C
    \end{align*}
    Moreover, using \cref{lemma:decomp-shortest-opt}, we get:
    \begin{align*}
        \sum_{C \in A} r_C \leq 2 OPT
    \end{align*}

    Let $C \in A$, we want to prove that $l_C \leq (i+1) \cdot r_C$ which will conclude the proof. Let $H$ be the set of components whose compressed distance was strictly less than $l_C / (i + 1)$ at the time of merging. Because their compressed distance is strictly less than $l_C$, it implies they were not compatible with $C$ and therefore are even and of rank strictly less than $\rank(C)$. Moreover, during the regular merge, their nearby rank is set to be at least $i$. Finally, by design of the algorithm, we note that the maximum arrival time of any request in $H$ is at most $t_{max}(C) + l_C$. 
    
    We want to prove that all components within compressed distance $l_C / (i + 1)$ of $C$ in $\dec_{i-1}$ are even, this would imply that $r_C$, the shortest distance from $C$ to another odd component in $\dec_{i-1}$ is at least $l_C / (i + 1)$ away from $C$. Because the algorithm assigns all components in $H$ a nearby rank of at least $i$, \cref{lemma:nearby-after-edges,coro:compo-edge-rank} implies that components in $H$ have exactly matching components in $\dec_{i-1}$, which are therefore also even. We consider the shortest path between $C$ and its closest odd component in $\dec_{i-1}$, which has weight $r_C$. Let us consider the first edge along this path between a component in $\{C\}\cup H$ and a component outside of $\{C\}\cup H$. $P$ starts from $C$ and ends with an odd component different from $C$. Given that $H$ contains only even components, $P$ ends in a component not in $\{C\}\cup H$, so such an edge $e$ always exists. Let $e = (u,v)$, with $u$ being a request in $\{C\}\cup H$ and $v$ being a request outside of it. We have two cases:
    \begin{itemize}
        \item Request $v$ has already arrived when $C$ merged. Let $C'$ be the component $v$ belonged to when $C$ merged. Because $v$ is not a request in $H$, this means that $C' \notin H$ and by definition of $H$, at the time of merge, $D(C, C') \geq l_C / (i+1)$. However, we observe that the prefix of $P$ until $v$ is a path from $C$ to a request in $C'$ which only goes through components in $H$, which existed at the time of merge and were even. So by definition, this prefix of $P$ must have weight at least $l_C / (i+1)$, so $w(P) \geq l_C / (i + 1)$.
        \item Request $v$ arrived after $C$ merged. Let $t_{merge}$ be the time when $C$ merged. By design of the algorithm, we know that $t_{merge} \geq t_{max}(C) + 2l_C$. Moreover, using \cref{lemma:regular-invariants}, $\rank(C) < i$, so all components $C' \in H$ verify $D(C,C') < l_C / (\rank(C) + 2)$. As a consequence, the algorithm ensured in the combining step that all components $C' \in H$ satisfy $t_{max}(C') \leq t_{max}(C) + l_C$. Therefore:
        \begin{align*}
            w(e) &= d(u,v) \\*
            & \geq |t(v) - t(u)| \\
            &= t(v) - t(u) \\
            &\geq t_{merge} - \max (t_{max}(C), \max_{C' \in H} t_{max}(C')) \\
            &\geq t_{max}(C) + 2l_C - (t_{max}(C) + l_C) \\
            &\geq l_C
        \end{align*}
        Therefore, $w(P) \geq w(e) \geq l_C \geq l_C / (i + 1)$
    \end{itemize}
    In both cases, $r_C = w(P) \geq l_C / (i + 1)$, therefore $l_C \leq (i+1) \cdot r_C$ and:
    \begin{align*}
         w(R_i) &= \sum_{C \in A} l_C \\
         & \leq \sum_{C \in A} (i+1)\cdot r_C \\
         & = (i+1) \cdot \sum_{C \in A} r_C \\
         & \leq (i+1)\cdot 2OPT
    \end{align*}
\end{proof}

We now need to bound the weight of the special edges. As mentioned previously, we designed the algorithm in a way that special edges have a small cost compared to the existing size of the components. To prove this, we need re-partition the set of special edges $(S_i)_{i \in [\![0, \lfloor \log m\rfloor ]\!]}$ into a new partition $(S'_i)_{i \in [\![0, \lfloor \log m\rfloor ]\!]}$. Let $e$ be a special edge, it was added by the special merge procedure. We remark that this procedure is only called at two different locations by our algorithm. The first one being in the combining step, when merging a components $C_3 \in H$ into $C_1$ with edge rank $\rank(C_1)$. In this case, we add this edge to $S'_{\rank(C_1)}$, we remark that this edge is also part of $S_{\rank(C_1)}$. The second case is during the nearby fixup procedure, merging some component $C_1$ into $C_2$, where the rank or nearby rank of $C_2$ is at least $\nrank(C_1)$. In this case we add the edge to $S'_{\nrank(C_1)}$ and we remark that by design of the algorithm $\nrank(C_1) \leq \max (\rank(C_2), \nrank(C_2))$ which is the edge rank. We notice that in both cases, an edge originally in some $S_i$ with $i \geq 0$ is placed in a set $S'_k$ with $k \leq i$. 

For $i \in [\![0, \lfloor \log m\rfloor ]\!]$, we now define:
\begin{align*}
    \mathcal{F}'_{i} = \bigcup_{k=1}^i ( R_k \cup S'_k)
\end{align*}

Because of the previous point, we have $\mathcal{F}'_{\lfloor \log m \rfloor} = \mathcal{F}_{\lfloor \log m \rfloor} = \mathcal{F}$. Moreover, for all $i \in [\![0, \lfloor \log m\rfloor ]\!], \mathcal{F}_{i} \subseteq \mathcal{F}'_{i}$.

\begin{lemma}\label{lemma:bound-special}
    Let $i \in [\![1, \lfloor \log m\rfloor ]\!]$, then:
    \begin{align*}
        w(S'_i) \leq \left(w(\mathcal{F}'_{i-1}) + w(R_i) \right) / i
    \end{align*}
\end{lemma}

\begin{proof}
    Let $i \in [\![1, \lfloor \log m\rfloor ]\!]$, let $A = \mathcal{F}'_{i-1} \cup R_i \subseteq \mathcal{F}$. We consider a special merge from $C_{from}$ to $C_{to}$ which contributes a path $P$, being a shortest path from $C_{from}$ to $C_{to}$, to $S'_i$. Let $D$ be the set of even components merged along $P$ (excluding endpoints). We want to prove two properties: (i) the edges of $T(C_{from})$ and the spanning trees of components in $D$ are in A and (ii) $w(P) \leq (w(C_{from}) + w(D)) / i$. We cover the two cases where a special merge can happen. 

    The first case is during the combining step. Using the algorithm notation, let $C_1 = C_{to}$ be the first component considered and $l_{C_1}$ be the compressed distance to its nearest compatible component at the time of the special merge. We have $\rank(C_1) = i$, and $C_{from} = C_3$ satisfies $D(C_1, C_3) < l_{C_1} / (rank(C_1) + 2)$ and $t_{max}(C_3) \geq t_{max}(C_1) + l_{C_1}$. Because $D$ and $C_3$ are not compatible with $C_1$, this means that their rank is strictly less than $\rank(C_1) = i$, which is the rank used for the special merge. Using \cref{coro:compo-edge-rank}, this implies that $C_3$ and $D$ are made from edges in $\mathcal{F}_{i-1}$, so in $A$ because $\mathcal{F}_{i-1} \subseteq \mathcal{F'}_{i-1}$.

    We have $t_{max}(C_3) \geq t_{max}(C_1) + l_{C_1}$, so there exists a request $v \in V(C_3)$ such that $t(v) \geq t_{max}(C_1) + l_{C_1}$. We have $w(P) = D(C_1, C_3) < l_{C_1} / (\rank(C_1) + 2)$. Given $P$, a simple path on the components, we now describe how to get $P'$, a simple path on the requests. To do so, we consider the fact that each edge in $P$ starts and ends at a specific request within its incident component. To link these together, consider two consecutive edges in $P$: the request where the first edge ends and the request where the second edge begins belong to a same component $C$. We can therefore connect these two requests using the simple path between them in the spanning tree $T(C)$. We end up with a path $P'$ from some $u \in C_1$ to $v$ made of two disjoint sets of edges: the first $E_1$ consists of the edges between requests in different components. We note that $w(E_1) = w(P)$. The second $E_2$ consists of edges between requests in the same component. By construction, these edges belong to $T(C_3)$ or to the spanning tree of components in $D$. Moreover, each edge in the spanning tree of a component is only used at most once: $P$ is a shortest path and does not pass through a component twice. Therefore $w(E_2) \leq w(C_3) + w(D)$. Because $d$ satisfies the triangle inequality and $P'$ is a path from $u$ to $v$, we get:
    \begin{align*}
        w(P') &\geq d(u,v) \\
        & \geq |t(v) - t(u)| \\
        & \geq t(v) - t(u) \\
        & \geq t_{max}(C_1) + l_{C_1} - t_{max}(C_1) \\
        & \geq l_{C_1} \\
        & \geq w(P) \cdot (\rank(C_1) + 2) \\
        & = w(P) \cdot (i+2)
    \end{align*}
    We also have:
    \begin{align*}
        w(P') &= w(E_1) + w(E_2) \\
        & \leq w(P) + (w(C_3) + w(D)) 
    \end{align*}
    From these two inequalities, we get the desired property $w(P) \leq (w(C_3) + w(D)) / (i+1) \leq  (w(C_3) + w(D)) / i =  (w(C_{from}) + w(D)) / i$.

    The second case is in the nearby fixup procedure, using the algorithm notation, when merging a component $C_1 = C_{from}$ to its closest component $C_2 = C_{to}$ having a rank of at least $\nrank(C_1)$, or nearby rank of at least $\nrank(C_1) + 1$. We remark that $i = \nrank(C_1)$ and that this special merge can happen in the middle of an iteration, meaning some invariants may not hold. Because components in $D$ are closer to $C_1$ than $C_2$, it implies that their rank is strictly less than $\nrank(C_1) = i$. Given that components in $D$ were not touched by the algorithm in the current iteration so far, we can use \cref{coro:compo-edge-rank}, which implies that they are made from edges in $\mathcal{F}_{i-1}$, so in $A$ because $\mathcal{F}_{i-1} \subseteq \mathcal{F'}_{i-1}$. We now consider $C_1$. The fixup procedure might have already performed some special merge before this one. However, we remark that the nearby rank of $C_1$ increases by at least $1$ during each fixup iteration. Therefore, each special merge in previous iterations of the fixup procedure contributed to $S'_k$ for some $k < i$, which is in $A$. Moreover, the regular merge right before the call to the fixup procedure had rank at most $\nrank(C_1) = i$, so the edges added were part of $R_k$ for some $k \leq i$, which is also in $A$. Finally, for the remaining edges, which were already present before the beginning of the algorithm iteration, we can use our invariants from \cref{lemma:nearby-greater,coro:compo-edge-rank}: given that $\rank(C_1) < \nrank(C_1) = i$, all remaining edges are part of $\mathcal{F}_{i-1}$.

    We consider the time $t'$ at which $C_1$ got assigned $i$ as its nearby rank. At time $t'$, $C_1$ was possibly smaller and as such, we call $C_1'$ the component $C_1$ as it was at time $t'$, we have $C'_1 \subseteq C_1$. Getting a non-$\bot$ nearby rank can only happen during a regular merge of some component $C'_{from}$ into a component $C'_{to}$ at rank $i$. Let $l$ be the distance from $C'_{from}$ to $C'_{to}$ at this time, $l = D_{t'}(C'_{from}, C'_{to})$. Because $C'_1$ was given a nearby rank, it means that $D_{t'}(C'_{from}, C'_1) < l / (i + 1)$.

    We want to prove that while $C'_1$ has not been merged, there is always a component of rank at least $\nrank(C_1)$ with distance at most $l / (i + 1)$ from $C'_1$. This, in particular, implies that $D(C_1,C_2) \leq  l / (i + 1)$. Using \cref{lemma:regular-invariants}, we consider the shortest path $P$ from $C'_1$ to $C'_{from}$ at time $t'$. We remark that it only passes through components with distance less than $l / (i + 1)$ from $C'_{from}$, hence all of these components' nearby rank was set to at least $i$. We observe using \cref{lemma:nearby-after-edges} that after $t'$, these components along the path are either the same or merged into a component of rank at least $i$. Moreover, after $t'$, $C'_{from}$ has itself merged into a component of rank at least $\nrank(C_1)$. Therefore, going along $P$ after $t'$, we always find a component of rank at least $i$ after some even components. By definition of $C_2$, we get that $D(C_1, C_2) \leq w(P) < l /(i+1)$.

    We now want to give a lower bound on $w(C_1)$. To be more precise, we want to prove that $w(C_1) \geq l \cdot(1 - 1/(i+1))$. We remark that in the nearby fixup procedure, $C_1$ is odd. Hence, using $\cref{lemma:request-in-odd}$, there exists a request $v \in V(C_1)$ which was always part of an odd component until $t$. We consider the state of this request at time $t'$:
    \begin{itemize}
        \item Request $v$ already arrived by $t'$ ($t(v) < t'$). But then, the component $v$ was in at time $t'$ was an odd component $C_v$ and as a consequence was compatible with $C'_{from}$. $C'_{to}$ being the closest compatible component at this time, this means that $D_{t'}(C'_{from}, C_v) \geq D_{t'}(C'_{from}, C_{to}) = l$. Because $C'_1$ was an even component at time $t'$, we get that:
        \begin{align*}
            D_{t'}(C'_1, C_v) &\geq D_{t'}(C'_{from}, C_v) - D_{t'}(C'_{from}, C'_1) \\
            &\geq l - l/(i+1) = l \cdot(1 - 1/(i+1))
        \end{align*}
        Let $u \in C'_1$, using \cref{lemma:comp-distance}, we have therefore $d(u,v) \geq D(C'_1, C_v) \geq l \cdot(1 - 1/(i+1))$. We remark that we chose $v$ such that it was in $V(C_1)$. As a consequence, both $u$ and $v$ are in $V(C_1)$. Because $d$ satisfies the triangle inequality, the spanning tree $T(C_1)$ on $V(C_1)$ satisfies $w(T(C_1)) \geq d(u,v) \geq l \cdot(1 - 1/(i+1))$.
        \item Request $v$ arrived after $t'$ ($t(v) \geq t'$). We note that by design of the algorithm, $t' \geq t_{max}(C'_{from}) + 2D_{t'}(C'_{from}, C_{to}) = t_{max}(C'_{from}) + 2l$. Moreover, using \cref{lemma:regular-invariants}, $\rank(C'_{from}) < i$, so $C_1'$ satisfies $D_{t'}(C'_{from}, C'_1) < l / (i + 1) \leq l / (\rank(C'_{from}) + 2)$. As a consequence, the algorithm ensured in the combining step that $t_{max}(C'_1) \leq t_{max}(C'_{from}) + l$. Let $u \in C_1'$, we therefore have:
        \begin{align*}
            d(u,v) &\geq |t(v) - t(u)| \\
            & \geq t(v) - t(u) \\ 
            & \geq t_{max}(C'_{from}) + 2l - (t_{max}(C'_{from}) + l) \\
            & = l
        \end{align*}
        We remark that both $u$ and $v$ are requests in $C_1$ and use the same argument as before: because $d$ satisfies the triangle inequality, the spanning tree $T(C_1)$ verifies:
        \begin{align*}
            w(C_1) \geq d(u,v) \geq l > l \cdot(1 - 1/(i+1))
        \end{align*}
    \end{itemize}
    From the previous points, we get that:
    \begin{align*}
        W(P) &= D(C_1, C_2) \\
        &\leq 1 /(i+1) \cdot l \\
        & \leq 1/(i+1) \cdot (1 - 1/(i+1))^{-1} \cdot w(C_1) \\
        & = w(C_1) / i \\
        &\leq  (w(C_{from}) + w(D)) / i
    \end{align*}

    We showed that (i) the edges of $T(C_{from})$ and the spanning trees of components in $D$ are in A and (ii) $w(P) \leq (w(C_{from}) + w(D)) / i$. We remark that the resulting component has at least an edge in $S'_i$ and therefore, because of (i), cannot be used again as $C_{from}$ or in $D$ for a new special merge contributing to $S'_i$. Hence, each edge in $A$ is used at most once as $C_{from}$ or $D$ for a merge contributing to $S'_i$. Moreover, using (i), these components are made only using edges from $A$. Thus, using (ii), summing over all merges which contribute to $S'_i$, we get $w(S'_i) \leq w(A) / i = \left(w(\mathcal{F}'_{i-1}) + w(R_i) \right) / i $.
\end{proof}

\begin{theorem} \label{theorem:bound-forest-weight}
    \begin{align*}
        w(\mathcal{F}) \leq OPT  \cdot (2 + o(1)) \log^2m
    \end{align*}
\end{theorem}

\begin{proof}
    We consider $H_k = \sum_{i=1}^k 1/k$ the harmonic series. We prove by induction for $i \in [\![0, \lfloor \log m\rfloor ]\!]$ that:
    \begin{align*}
        w(\mathcal{F}'_i) \leq 2 \cdot (i+1) \cdot (i + H_i)\cdot OPT
    \end{align*}
    
    \textit{Initialization}: For $i = 0$, $\mathcal{F}'_0$ contains no edge, so $w(\mathcal{F}'_0) = 0$. We remark that the right term is also $0$, so the property holds.

    \textit{Induction}: Let $i \in [\![1, \lfloor \log m\rfloor ]\!]$, we assume that the property is true for $w(\mathcal{F}'_{i-1})$. We remark that $\mathcal{F}'_{i} = R_i \cup S'_i \cup {F}'_{i-1}$. Using \cref{lemma:bound-special,lemma:bound-regular}, we get that $w(R_i) \leq 2i \cdot OPT$ and $w(S'_i) \leq \left(w(R_i) + w(\mathcal{F}'_{i-1})\right)/ i$. Therefore, using the induction hypothesis:
    \begin{align*}
        w(\mathcal{F}'_{i}) &\leq w(R_i) + w(S'_i) + w(\mathcal{F}'_{i-1}) \\
        & \leq (1 + 1/i) \cdot w(R_i) + (1 + 1/i) \cdot w(\mathcal{F}'_{i-1})\\
        & \leq 2 \cdot (1+1/i) \cdot (i+1) \cdot OPT +   2 \cdot (1+1/i) \cdot i \cdot (i - 1 + H_{i-1}) \cdot OPT \\
        & = 2 \cdot (i+1) \cdot (1+1/i) \cdot OPT +  2 \cdot (i+1) \cdot (i - 1 + H_{i-1}) \cdot OPT \\
        & = 2 \cdot (i+1) \cdot (1 + 1/i + i - 1 + H_{i-1}) \cdot OPT \\
        & = 2 \cdot (i+1) \cdot (i + H_i)\cdot OPT
    \end{align*}

    \textit{Conclusion}: We have:
    \begin{align*}
        w(\mathcal{F}) = w(\mathcal{F}'_{\lfloor \log m\rfloor}) \leq 2 \cdot (\lfloor \log m\rfloor+1) \cdot (\lfloor \log m\rfloor + H_{\lfloor \log m\rfloor})\cdot OPT
    \end{align*}
    
    Using the property of the harmonic series $H_k = \Theta(\log k) = o(k)$, we get $w(\mathcal{F}) \leq OPT  \cdot (2 + o(1)) \log^2m$.

\end{proof}

\subsection{Bounding the algorithm's cost}

In the previous part, we were able to bound the weight of the forest created by our algorithm. Using this bound, we will show how to establish an upper bound on the time cost before requests join the greedy procedure, then on the time cost and connection cost inside the greedy part.

Let $u$ be a request. We consider $t_a(u) = t(u)$ its time of arrival, $t_j(u)$ the time the request joined a greedy instance and $t_m(u)$ the time it got matched with another request. We remark that the total time cost for this request is $|t_m(u) - t_a(u)|$. Let $\delta_c(u) = |t_j(u) - t_a(u)|$ the time spent waiting before joining a greedy instance and $\delta_g(u) = |t_m(u) - t_j(u)|$ the time spent waiting after joining a greedy instance. Finally, we denote $\Delta_c = \sum_{u \in V_{req}} \delta_c(u)$ and $\Delta_g = \sum_{u \in V_{req}} \delta_g(u)$. We remark that the total time cost of the algorithm is $\Delta_c + \Delta_g$. We will first focus on $\Delta_c$, the total time spent by requests before joining a greedy instance. 

During the execution of the algorithm, for a given time $t$, we denote by $\comp_{odd}(t)$ the set of odd components present at time $t$. We give the following result:

\begin{lemma}\label{lemma:bound-delta-odd}
    \begin{align*}
        \Delta_c = \int_{t \geq 0} |\comp_{odd}(t)|dt
    \end{align*}
\end{lemma}

\begin{proof}
    As soon as a component has two or more requests not in a greedy instance, it makes them join its own greedy instance two at a time. From this, we conclude that at any given time, even components have no requests not inside a greedy instance while odd components have exactly one. The lemma follows from this observation.
\end{proof}

We can now explain the purpose of our wait tree pruning procedure. We recall that when considering an odd component $C_1$ in the merging procedure, if its closest compatible component $C_2$ is an odd component of strictly lower rank and $t \geq t_{max}(C_1) + 2D(C_1,C_2)$, then the algorithm does nothing and we add a waiting edge from $C_1$ to $C_2$. We say that $C_1$ is a \textbf{waiting} component, let $\comp_w(t)$ be the set of waiting components at a given point in time. Odd components which are not waiting are called \textbf{active} components. We denote $\comp_a(t)$ the set of active components. We note that at any time $t$, $\comp_{odd}(t) = \comp_w(t) \cup \comp_a(t)$.

\begin{lemma}\label{lemma:compare-active-waiting}
    At any time $t \geq 0$, $|\comp_w(t)| \leq \log m \cdot |\comp_a(t)|$
\end{lemma}

\begin{proof}
    At any time $t$, in the algorithm we consider $\mathcal{W}(t)$ the oriented graph whose vertices are the odd components $\comp_{odd}(t)$ and edges are the waiting edges. We call $\mathcal{W}(t)$ the \textit{waiting forest}. Because a waiting edge goes from one component to another with strictly lower rank, it cannot have cycles. Moreover, each odd component waits on at most a single component, so has out-degree at most $1$ in $\mathcal{W}(t)$. From these observations, we conclude that $\mathcal{W}(t)$ is indeed a forest. Moreover, each oriented tree $T$ in $\mathcal{W}(t)$ has its edges oriented towards the root, which has out-degree $0$. 

    We also remark that a component is active if and only if its out-degree in $\mathcal{W}(t)$ is $0$. Therefore, all trees in $\mathcal{W}(t)$ consist of a single active component being its root, and some waiting components.

    We claim that each tree in $\mathcal{W}(t)$ has size at most $\lfloor \log m \rfloor + 1$. The reason is that, using \cref{coro:max-comp-rank}, the maximum component rank is $\lfloor \log m \rfloor$. So if a tree has a greater size, using the pigeonhole principle, we could find two components in this tree with the same rank, but the tree pruning procedure would immediately merge them together. Thus, we have:
    \begin{itemize}
        \item A graph where each active and waiting components appear as vertices.
        \item Each connected component of this graph contains exactly one active component.
        \item Each connected component of this graph has size at most $\lfloor \log m \rfloor + 1$
    \end{itemize}
    Using these properties together, we get that $|\comp_w(t)| \leq \lfloor \log m \rfloor \cdot |\comp_a(t)| \leq \log m \cdot |\comp_a(t)|$
\end{proof}

Using the lemma above, we can now focus on active components only. Indeed, any result on active components would apply to all odd components with an additional logarithmic factor.

\begin{lemma} \label{lemma:bound-active-forest}
    \begin{align*}
        \int_{t \geq 0} |\comp_{a}(t)|dt \leq 8 \cdot \log m \cdot w(\mathcal{F})
    \end{align*}
\end{lemma}

\begin{proof}
We now want to bound the time an odd component stays active. Let $r \in [\![0, \lfloor \log m\rfloor ]\!]$, we consider a component $C(t)$ of rank $r$ evolving from the time $t_{re}(C, r)$ where it reached rank $r$ to the time $t_{le}(C, r)$ where either: it merges into a bigger component, its rank increased, or the time the algorithm ends if this component is even and never gets used again. Using \cref{lemma:component}, we can find $C' \in \dec_r$ such that $C(t_{le}(C,r)) = C'$.

We will divide the time $C(t)$ stays active in two and give an upper bound for each of these intervals. We say that $C(t)$ stabilizes at rank $r$ if no more odd component merges into $C(t)$ while its rank is $r$. We denote by $t_{st}(C, r)$ the stabilization time of $C(t)$ at rank $r$. We have $t_{re}(C,r) \leq t_{st}(C,r) \leq t_{le}(C,r)$. We will give an upper bound on the amount of time a component can stay active before its stabilization time and after it at a given rank.

Assume $r \geq 1$, $C'$ is a component in $\dec_r$ where components are defined as connected components in the forest $\mathcal{F}_r$. Because $\mathcal{F}_{r-1} \subseteq \mathcal{F}_r$, we can find a set of components $B \subseteq \dec_{r-1}$ such that $V(C') = \cup_{C" \in B} V(C")$, i.e $C'$ is made from components of $B$ merging together. Using \cref{lemma:component}, we can show that every component $C"$ in $B$ existed at some time $t'$, different for each component, right before it merged into $C(t')$. Let $C_1, ..., C_k$ be the odd components of $B$. Using \cref{lemma:request-in-odd}, we can find some requests $v_1 \in C_1, ..., v_k \in C_k$ such that each of those requests was in an odd component until it got merged into $C$. We will show that $C$ can remain active at most $3w(C')$ units of time between $t_{re}(C,r)$ and $t_{st}(C,r)$.

Let $t_1 = t_{re}(C,r)$ and $t_2 = t_{st}(C,r)$. We consider $u = repr(C(t_1))$ which by design stays the same at any time and is already part of $C(t_1)$. We consider:
\begin{align*}
    l = \max_{i \in [\![1, k ]\!]} d(u, v_i)
\end{align*}
We claim that after $t_1 + w(C') + 2l$ and before $t_2$, $C(t)$ cannot be active. To do so, consider some time $t$ such that $t > t_1 + w(C') + 2l$ and $t < t_2$, let us show that $C(t)$ is even or waiting. If $C(t)$ is even, we are done, so let us assume that $C(t)$ is odd. Let $l'$ be the distance from $C(t)$ to its closest compatible component, we claim that $l' \leq l$. The reason is that because $t < t_2$, the component did not stabilize yet, so there must exist some $C_i$ such that $C_i$ did not merge into $C(t)$ yet. Therefore, by definition, at time $t$, $v_i$ is in an odd component $\tilde{C}$ distinct from $C(t)$. Because this component is odd, it is compatible with $C(t)$ and using \cref{lemma:comp-distance}, $D_t(C(t), \tilde{C}) \leq d(u, v_i) \leq l$, so by minimality $l' \leq l$. Moreover, there must exist some $u' \in V(C(t))$ such that $t_{max}(C(t)) = t(u')$. Because $u$ and $u'$ are both in $C'$ and $d$ satisfies the triangle inequality, we have:
\begin{align*}
    t_{max}(C(t)) - t(u) &= t(u') - t(u) \\
    & \leq d(u', u) \\
    & \leq w(C')
\end{align*}
Because $u = repr(C(t_1))$, we have $t(u) \leq t_1$, so $w(C') \geq t_{max}(C(t)) - t_1$. Combining these inequalities, we get:
\begin{align*}
    t &> t_1 + w(C') + 2l \\
    &\geq t_1 + t_{max}(C(t)) - t_1 + 2l' \\
    &\geq t_{max}(C(t)) + 2l'
\end{align*}
Because $t$ is after $t_{max}(C(t)) + 2l'$, where $l'$ is the distance from $C(t)$ to its closest compatible component, and no merge happened yet, it means that $C(t)$ is waiting. Moreover, we note that $l$ is the distance between $u$ and some $v_i$ and both $u$ and $v_i$ are in $V(C')$. Because $d$ satisfies the triangle inequality and $T(C')$ is a spanning tree over $V(C')$, we get that $l \leq w(C')$. Therefore, between $t_1$ and $t_2$, $C(t)$ can only be active until time $\min(t_2, t_1 + 3w(C'))$ and thus can remain active at most $3w(C')$ units of time. We remark that components with $r = 0$ are singletons and therefore, $t_2 = t_1$ hence this property holds immediately in this case.

We now consider the time $C$ spends active between $t_{st}(C,r)$ and $t_{le}(C,r)$. Let $t_1 = t_{st}(C,r)$ and $t_2 = t_{le}(C,r)$. By our choice of $t_1$, being the stabilization time, only even component will merge with $C$ in this time frame, so $C$ will keep the same parity, and in particular the same one as $C(t_2) = C'$. If $C'$ is even, $C$ will remain even, so not active between $t_1$ and $t_2$. We now assume that $C'$ is odd.

Because all final components are even, it follows that at time $t_2$, $C(t_2)$ will merge into a bigger component or its rank will increase, so we can focus on this case only. However, it turns out that bounding the time between $t_1$ and $t_2$ for a single component cannot be done efficiently. Instead, we will do so by considering a component and bounding together the time spent active after stabilization of all lower rank components which merge into it.

We thus consider a new $r > 0$ and $C \in \dec_r$. As explained above, we will take a new approach and instead bound the active time after stabilization of all odd components which merge into $C$. Because $\mathcal{F}_{i-1} \subseteq \mathcal{F}_{i}$, we can consider the sub-components of $C$ in $\dec_{r-1}$. Let $C'_1, ..., C'_k$ be the set of all odd such sub-components. Using \cref{lemma:request-in-odd}, for each $C'_i$, we can find a request $v_i$ such that $v_i$ was always part of an odd component until it merged into $C$. We observe at any time before merging into $C$, the component containing $v_i$ was odd and thus compatible with any other component.

We note that $v_1, .., v_k$ are requests in the spanning tree $T(C)$, we consider an Euler tour $E_T$ from $T(C)$ where we only keep the vertices in $\{v_1, ..., v_k\}$. Because we are in a metric space, $w(E_T) \leq 2w(T(C)) = 2w(C)$. We orient the tour $E_T$ in an arbitrary way and for each $v_i$, we consider $l_i$ to be the distance from $v_i$ to its right neighbor along $E_T$. By construction, $\sum_{i=1}^k l_i = w(E_T) \leq 2w(C)$. 

We now consider an odd component $C'_i$ among these sub-components. Let $C_i(t)$ be the component representing it during the algorithm execution. Let $t_1 = t_{st}(C_i, \rank(C'_i))$ and $t_2 = t_{le}(C_i, \rank(C'_i))$. We want to give an upper bound on the time $C_i(t)$ stays active after stabilization and until it merges into $C$. Because no odd component joins $C_i$ after $t_1$ and $v_i$ is always part of an odd component until it joins $C_i$, we conclude that $v_i$ joined $C'_i$ before $t_1$, and $t(v_i) \leq t_1$.  With the same argument as before, we get that for all $t \geq t_1, \ t_{max}(C_i(t)) \leq t(v_i) + w(C'_i) \leq t_1 + w(C'_i)$. We claim that between $t_1 + w(C'_i) + 2l_i$ and $t_2$, component $C_i(t)$ will always be waiting. We consider some time $t$ such that $t > t_1 + w(C'_i) + 2l_i$ and $t < t_2$. Because $t$ is after the stabilization time $t_1$, the parity of $C'_i$ does not change so $C_i(t)$ is odd. We note that by definition, $l_i$ is the distance between $v_i \in C_i(t)$ and some $v_j$ in another component which is compatible with $C_i(t)$. Therefore, using \cref{lemma:comp-distance}, the compressed distance between $C_i(t)$ and its closest compatible component $l'$ is at most $l_i$. We remark that $t > t_1 + w(C'_i) + 2l_i \geq t_{max}(C_i(t)) + 2l'$, therefore if the component was not waiting, it would have already merged into its closest component, so it is waiting. Thus after $t_1$, component $C_i(t)$ is active at most $w(C'_i) + 2 l_i$ units of time.

We now compute the total active time after stabilization of odd components in $B$. Using the previous part, it is at most:
\begin{align*}
    \sum_{i = 1}^k (w(C'_i) + 2l_i)
\end{align*}
We remark that $\sum_{i=1}^k l_i \leq 2w(C)$ and all $T(C'_i)$ are disjoint and a subset of $T(C)$, so $\sum_{i=1}^k w(C'_i) \leq w(C)$. Therefore, the total active time for components which merge into $C$ after stabilization is at most $5w(C)$.

We can finally compute the total active time for all components, which is the sum of the active time before stabilization and after stabilization for each component at each rank. We also note that $w(C) = 0$ for a component of rank $0$, as they are singletons:
\begin{align*}
    \int_{t \geq 0} |\comp_{a}(t)|dt &\leq \sum_{r = 0}^{\lfloor \log m \rfloor} \sum_{C \in \dec_r} 3w(C) + \sum_{r = 1}^{\lfloor \log m \rfloor} \sum_{C \in \dec_r} 5w(C) \\
    & = 8 \sum_{r = 1}^{\lfloor \log m \rfloor} \sum_{C \in \dec_r} w(C) \\
    & = 8 \sum_{r = 1}^{\lfloor \log m \rfloor} w(\mathcal{F}_r) \\
    & \leq 8  \sum_{r = 1}^{\lfloor \log m \rfloor} w(\mathcal{F}) \\
    & \leq 8 \cdot \log m \cdot w(\mathcal{F})
\end{align*}

\end{proof} 

\begin{theorem}\label{theorem:bound-waiting-outside}
    \begin{align*}
        \Delta_c \leq (16 + o(1)) \log^4 m \cdot OPT
    \end{align*}
\end{theorem}

\begin{proof}
    We use the fact that at any time, $\comp_{odd}(t) = \comp_w(t) \cup \comp_a(t)$. Therefore, $|\comp_{odd}(t)| \leq  |\comp_w(t)| + |\comp_a(t)|$. We have:
    \begin{align*}
        \Delta_c &= \int_{t \geq 0} |\comp_{odd}(t)|dt && \text{from \cref{lemma:bound-delta-odd}} \\
        & \leq \int_{t \geq 0} |\comp_{odd}(t)|dt \\
        & \leq \int_{t \geq 0} (|\comp_w(t)| + |\comp_a(t)|)dt \\
        & \leq \int_{t \geq 0} (\log m |\comp_a(t)| + |\comp_a(t)|)dt && \text{from \cref{lemma:compare-active-waiting}}\\
        & \leq (1 + o(1)) \log m \int_{t \geq 0} |\comp_a(t)|dt \\
        & \leq (1 + o(1)) \log m \cdot 8 \cdot \log m \cdot w(\mathcal{F}) && \text{from \cref{lemma:bound-active-forest}} \\
        &\leq (8 + o(1))  \log^2 m \cdot w(\mathcal{F}) \\
        &\leq (16 + o(1)) \log^4 m \cdot OPT && \text{from \cref{theorem:bound-forest-weight}}
    \end{align*}
\end{proof}

We can finally tackle the last part of the algorithm which is the greedy approximation. To do so, we need to approximate the size of a tour over all requests in each greedy invocation to be able to use \cref{theorem:greedy-upper-bound}. We denote by $M_g$ the connection cost of the algorithm. We observe that the total cost of the algorithm is therefore $\Delta_c + \Delta_g + M_g$. We also note that all matchings are done within greedy instances.

\begin{theorem}\label{theorem:bound-waiting-inside}
    \begin{align*}
        \Delta_g + M_g \leq \left(80 + o(1)\right) \log^5 m \cdot OPT
    \end{align*}
\end{theorem}

\begin{proof}
We define $\dec = \bigcup_{r=0}^{\lfloor \log m \rfloor} \dec_r$ the multiset of all possible components in our component decompositions. We consider the greedy invocation associated with a request $u \in V_{req}$, we remark that as long as the component containing $u$ does not merge into another component, its greedy invocation can receive requests from within the component. Let $C$ be the component with representative $u$ right before it got merged into another component or at the end of the algorithm if this never happens. Using \cref{lemma:component}, we note that $C$ corresponds to a component in $\dec_r$ for some $r \geq 0$. Moreover, a component having a single representative, this component in $\dec_r$ cannot correspond to another component with a different representative. Therefore, for each request $u \in V_{req}$, we can find a distinct component $C$ in $\dec$ such that the greedy instance on $u$ only matches requests in $V(C)$. For a component $C \in \dec$, we denote by $\req(C)$ the set of requests that are processed by the greedy instance associated with $C$. If no request got attached to the greedy instance on $C$, then $\req(C) = \emptyset$. We remark that for all $C \in \dec$, $\req(C) \subseteq V(C)$.

We consider a component $C \in \dec$ and want to bound the cost of the greedy associated with this component. We note that, by design of the algorithm, requests join the greedy algorithm with a different arrival time compared to their original arrival time. In particular, the arrival time of a request in the greedy invocation is $t_j(u)$ instead of $t_a(u)$. $T(C)$ being a spanning tree over $V(C)$, we can consider an Euler tour $E_C$ from the spanning tree $T(C)$ where we then remove requests which are not in $\req(C)$. Therefore, $E_C$ is a tour on $\req(C)$. Because we are in a metric space, $w(E_C) \leq 2w(C)$. We now consider $w'$ and $d'$ to be the modified edge weight and distance taking into account the time of joining the greedy algorithm $t_j(u)$ instead of the time of arrival $t_a(u)$ for a request $u$. We observe that by design of our algorithm, each greedy instance uses $d'$ instead of $d$ for the distance. For an edge $e = (u,v)$, we have:
\begin{align*}
    w'(e) &= d'(u,v) \\
    & = d((x(u), t_j(u)), (x(v), t_j(v))) \\
    & \leq d((x(u), t_j(u)), (x(u), t_a(u))) + d((x(u), t_a(u)), (x(v), t_a(v))) + d((x(v), t_a(v)), (x(v), t_j(v))) \\
    & = t_j(u) - t_a(u) + d(u,v) + t_j(v) - t_a(v) \\
    & = \delta_c(u) + \delta_c(v) + d(u,v)
\end{align*}
We note that $E_C$ is a tour of $\req(C)$, so each request in $\req(C)$ appears as the endpoint of exactly two edges in $E_C$, therefore:
\begin{align*}
    w'(E_C) \leq w(E_C) + 2\sum_{v \in \req(C)} \delta_c(v) \leq 2 w(C) + 2\sum_{v \in \req(C)} \delta_c(v)
\end{align*}
Moreover, $E_C$ is a valid traveling salesman tour for $\req(C)$, therefore the optimal traveling salesman tour for requests handled by the greedy invocation on $C$ has weight at most $2 w(C) + 2\sum_{v \in \req(C)} \delta_c(v)$. We denote by $M_C$ and $W_C$ the communication cost and matching cost of the greedy instance in $C$ respectively, we can now apply \cref{theorem:greedy-upper-bound}:
\begin{align*}
    W_C + M_C &\leq \frac{5}{2}(\lceil \log |\req(C)| \rceil + 1) \left( 2 w(C) + 2\sum_{v \in req(C)} \delta_c(v) \right) \\
    & \leq 5 (\lceil \log m \rceil + 1) \left( w(C) + \sum_{v \in req(C)} \delta_c(v) \right)
\end{align*}

We note that each request appears only in the greedy instance of exactly a single component, therefore $\displaystyle \sum_{C \in \dec} \sum_{v \in \req(C)}\delta_c(v) = \sum_{v \in V_{req}} \delta_c(v)$.
We can now bound the total wait time and communication cost, given that each request is matched in a single greedy instance:

\begin{align*}
    \Delta_g + M_g &= \sum_{C \in \dec} \left( W_C + M_C \right) \\
    & \leq \sum_{C \in \dec} \left( 5 (\lceil \log m \rceil + 1) \left( w(C) + \sum_{v \in \req(C)} \delta_c(v) \right) \right) \\
    & = (5 + o(1)) \log m \left( \sum_{C \in \dec} w(C) + \sum_{C \in \dec} \sum_{v \in \req(C)}\delta_c(v) \right) \\
    & = (5 + o(1)) \log m \left( \sum_{r = 1}^{\lfloor \log m \rfloor} \sum_{C \in \dec_r} w(C) + \sum_{v \in V_{req}} \delta_c(v) \right) \\ 
    & = (5 + o(1)) \log m \left( \sum_{r = 1}^{\lfloor \log m \rfloor} w(\mathcal{F}_r) + \Delta_c \right) \\
    & \leq (5 + o(1)) \log m \left( \sum_{r = 1}^{\lfloor \log m \rfloor} w(\mathcal{F}) + \Delta_c \right) \\
    & \leq (5 + o(1)) \log m \cdot \left( \log m \cdot w(\mathcal{F}) + \Delta_c \right) \\
    & \leq (5 + o(1)) \log m \cdot ( (2 + o(1)) \log^3 m \cdot OPT  \\
        &  \ \ \ \ \ \ \ \ + (16 + o(1)) \log^4 m \cdot OPT ) && \text{Using \cref{theorem:bound-forest-weight,theorem:bound-waiting-outside}} \\
    & = (80 + o(1)) \log^5 m \cdot OPT
\end{align*}
    
\end{proof}

\begin{theorem}
    \textsc{ONLINE\_MATCHING} has competitiveness $(80 + o(1)) \log^5 m$.
\end{theorem}

\begin{proof}
    The total cost of the algorithm is:
    \begin{align*}
        \Delta_c + \Delta_g + M_g &\leq (16 + o(1)) \log^4 m \cdot OPT +  (80 + o(1)) \log^5 m \cdot OPT && \text{Using \cref{theorem:bound-waiting-outside,theorem:bound-waiting-inside}} \\
        & = (80 + o(1)) \log^5 m \cdot OPT
    \end{align*}
    Because $OPT$ is the optimal cost an algorithm knowing in advance the arrival of every request achieves, \textsc{ONLINE\_MATCHING} has indeed competitiveness $(80 + o(1)) \log^5 m$.
\end{proof}
\section{Conclusion}

In this paper, we describe the first algorithm with polylogarithmic competitiveness for the general online min-cost perfect matching with delays problem, with the only assumption being that the requests lie in a metric space. This algorithm is an exponential improvement over previous results and is only a polylogarithmic factor away from the lower bound $\Omega(\log m / \log \log m)$.

Several open problems remain. The primary open problem is to tighten the gap with the lower bound, either by developing better algorithms or improving the lower bound itself. We observe that the lower bound holds even for randomized algorithms, while ours is deterministic. As such, one can wonder if randomized algorithms can offer an improvement for this problem. We also remark that our algorithm is somewhat complex. We believe that a similar approach could quite likely yield a simpler algorithm and possibly remove one or two logarithmic factors from the competitive ratio.

Finally, it is natural to ask whether this algorithm can be adapted for the online min-cost bipartite perfect matching with delays problem. Although previous works typically provided variants for both the general and bipartite settings, several key parts of our approach, specifically the greedy algorithm and the component decomposition, do not naturally extend to the bipartite case.

%\bibliography{bib/refs, bib/abbrev3, bib/crypto_crossref}

\newpage
\appendix
\section{Missing proofs}

We provide here all the proofs omitted in the main body of the paper. 

\OnlineOfflineSame*

\begin{proof}
    We consider a solution $(p_i, q_i, t_i)_{i \in [\![1;m/2]\!]}$ to the \mpmd problem, for $i \leq m/2$, we observe that $t_i \geq \max(t(p_i), t(q_i))$, as a consequence:
    \begin{align*}
        |t_i - t(p_i)| + |t_i - t(q_i)| &\geq |\max(t(p_i), t(q_i)) - t(p_i)| + |\max(t(p_i), t(q_i)) - t(q_i)| \\
        & = |t(p_i) - t(q_i)|
    \end{align*}
    Therefore:
    \begin{align*}
        g(x(p_i), x(q_i)) + |t_i - t(p_i)| + |t_i - t(q_i)| \geq g(x(p_i), x(q_i)) + |t(p_i) - t(q_i)| = d(p_i, q_i)
    \end{align*}

    As a consequence:
    \begin{align*}
        \sum_{i=1}^{m/2} g(x(p_i), x(q_i)) + |t_i - t(p_i)| + |t_i - t(q_i)| \geq \sum_{i=1}^{m/2} d(p_i, q_i)
    \end{align*}
    Conversely, if we consider a solution $(p_i, q_i)_{i \in [\![1;m/2]\!]}$ to the offline problem, we define $\forall i \leq m/2$, $t_i = \max(t(p_i), t(q_i))$. We observe that with this definition, $|t_i - t(p_i)| + |t_i - t(q_i)| = |t(p_i) - t(q_i)|$, thus:
    \begin{align*}
        \sum_{i=1}^{m/2} d(p_i, q_i) = \sum_{i=1}^{m/2} g(x(p_i), x(q_i)) + |t_i - t(p_i)| + |t_i - t(q_i)|
    \end{align*}
    From the first result, any online solution of cost $C$ yields an offline solution of cost at most $C$. From the second result, any offline solution of cost $C'$ yields an online solution of cost $C'$. Using these two properties together, we conclude that $OPT = OPT'$.
\end{proof}

\DecompShortestOpt*

\begin{proof}
    Let $M_{OPT}$ be the min-cost perfect matching of weight $W_{OPT}$. Let $E \subseteq M_{OPT}$ be the subset of these edges such that the two endpoints of any edge are in different components. By identifying a vertex to the component containing it, we can define $G' = (\comp, E)$ as an undirected graph on components. We note that this graph is not simple (there may be multiple edges between two components), but it has no self-loops. We consider a component $C \in \comp$, let $k$ be the number of internal edges (edges with both endpoints in $V(C)$) in $M_{OPT}$. Because every vertex is the endpoint of exactly one edge of $M_{OPT}$, we get that $\mathrm{deg}_{G'}(C) = |V(C)| - 2k$. In particular, the parity of $\mathrm{deg}_{G'}(C)$ is the same as the parity of $|V(C)|$. Therefore, a component in $G'$ has odd degree if and only if it is an odd component.

    Let $k_{odd}$ be the number of odd components in $\comp$, which is also the number of vertices with odd degree in $G'$. Using the handshaking lemma, $k_{odd}$ is even. We claim that we can find a set of $k_{odd}/2$ edge-disjoint paths $(P_i)_{i \in [\![1,k_{odd}/2]\!]}$ such that each path is between two distinct odd vertices in $G'$ and every odd vertex in $G'$ appears exactly once as the endpoint of one of those path. 
    
    To do so, we first prove that we can always find a path between two distinct vertices with odd degree in $G'$. If we consider $v_1$ a vertex with odd degree, it must have degree at least $1$ so we can find an edge $e_1 = (v_1, v_2)$. If $v_2$ has odd degree, then $e_1$ satisfies the property we look for as a path. Otherwise, we remove $e_1$ from $G'$, the degree of $v_1$ then becomes even while the degree of $v_2$ becomes odd. We can repeat this operation now starting with $v_2$ until we end up on a vertex with odd degree, this gives us the path we are looking for. 

    When we have such a path, we remove it from $G'$: the degree of its two endpoints becomes even and the degree of other vertices does not change. We can repeat this process $k_{odd}/2$ times and will be left with only vertices with even degree. From this, we get the set of edge-disjoint paths $(P_i)_{i \in [\![1,k_{odd}/2]\!]}$ described above.

    We consider one such path $P = C_1 \rightarrow C_2 \rightarrow \ldots \rightarrow C_l$, let $w(P)$ be the sum of the edges along such a path. According to our assumptions, $C_1$ and $C_l$ are distinct components. Let $s_1 < l$ be the last index $C_1$ appears along path $P$ and let $s_2 \in ]\!]s_1, l]\!]$ be the index after $s_1$ where an odd component distinct from $C_1$ appears. We remark that $C_{s_1} \rightarrow C_{s_1 + 1} \rightarrow \ldots \rightarrow C_{s_2}$ is a path between two distinct odd components where parts within even components are ignored. Thus by definition, the weight of this path is at least $D(C_1, C_{s_2})$. Therefore, $w(P) \geq D(C_1, C_{s_2}) \geq r_{C_1}$. By symmetry, $w(P) \geq r_{C_2}$.

    We remark that the paths $(P_i)_{i \in [\![1,k_{odd}/2]\!]}$ are edge-disjoint, therefore:
    \begin{align*}
        \sum_i w(P_i) &= w\left(\bigcup_i P_i \right) \\
        & \leq w(M_{OPT}) \\
        & = W_{OPT}
    \end{align*}
    Moreover, each odd component appears exactly once at the endpoint of a path in $(P_i)_i$ and each path has exactly two endpoints, therefore:
    \begin{align*}
        \sum_i 2 w(P_i) \geq \sum_{C \in \comp_{odd}} r_C
    \end{align*}
    Putting the two previous equations together, we get our result:
    \begin{align*}
        \sum_{C \in \comp_{odd}} r_C \leq 2 W_{OPT}
    \end{align*}
\end{proof}

\GreedyUpperBound*

\begin{proof}
    Let $M$ be the matching output by $M_{GREEDY}$. We first want to bound the space cost $w(M)$ of the matching. We will use \cref{lemma:metric-tsp} again. For $u \in V_{req}$ a request, let $v$ be the request \textsc{GREEDY\_ONLINE} matches it to. We define $l_u = d(u,v) / 2$. We show that it satisfies the properties required by the lemma:
    \begin{itemize}
        \item  Let $p,q$ be two distinct requests. If the greedy algorithm matched $p$ with $q$, then $l_p = l_q$ and $d(p,q) = 2 l_p \geq \min (l_p, l_q)$. Otherwise, without loss of generality, let us assume that $p$ got matched before $q$. We now consider the algorithm at the time it matched $p$ with some other vertex, there are two cases:
        \begin{itemize}
            \item $p$ was the first request considered ($u$ in the context of the algorithm). Let $v$ be the vertex matched to it. The algorithm waited past time $t(p) + d(p,v)$ to match it, which using \cref{lemma:arrival-time} implies that all requests within a radius $d(p,v)$ from $p$ already arrived. Because it decided to match $p$ with $v$ and not $q$ which was unmatched, it means that $q$ had not arrived yet or was further away from $u$ than $v$. In both cases, $d(p,q) \geq d(p,v)$, so $2l_p = d(p,v) \leq d(p,q)$. Hence $d(p,q) \geq 2l_p \geq \min (l_p, l_q)$.
            \item $p$ was the second request considered ($v$ in the context of the algorithm). The algorithm therefore matched $p$ to some other request $u$ at time $t$. We have $t \geq t(u) + 2 d(u,p)$. Because $|t(u) - t(p)| \leq d(u,p)$, it implies that $t \geq t(p) + d(u,p)$. Assume by contradiction that $d(p,q) < d(u,p)/2$. Then we would have $t > t(p) + 2 d(p,q)$, hence the algorithm would have already matched $p$ at time $t(p) + 2 d(p,q)$ with $q$ if it were not already matched. But we assumed that $p$ got matched at time $t > 2 d(p,q)$, hence the contradiction. Therefore $d(p,q) \geq d(u,p)/2 = l_p \geq \min(l_p, l_q)$.
        \end{itemize}
        \item Let $p \in V_{req}$, we assume $p$ got matched to some other request $q$ so $l_p = d(p,q)/2$. The TSP consists of two paths between $p$ and $q$. Given that we are in a metric graph, it follows that $OPT_{TSP} \geq 2 d(p,q)$ and therefore $l_p \leq \frac{1}{4} OPT_{TSP}$.
    \end{itemize}
    We can thus use \cref{lemma:metric-tsp}:
    \begin{align*}
        \sum l_p \leq \frac{1}{2}(\lceil \log n \rceil + 1) OPT_{TSP}
    \end{align*}
    We remark that by the definition of $(l_p)$, the weight of every matching edge appears exactly twice (one for each endpoint) in $(l_p)$, therefore $\sum 2 l_p = 2 w(M)$, so the overall connection cost satisfies:
    \begin{align*}
        w(M) \leq \frac{1}{2}(\lceil \log n \rceil + 1) OPT_{TSP}
    \end{align*}

    We now consider the time cost $\mathcal{W}_{GREEDY}$ of the algorithm. Let $u$ and $v$ be two requests matched together by the algorithm at time $t$. Without loss of generality, let us assume that $t(u) \leq t(v)$. Let $t' = t(u) + 2d(u,v)$. We want to prove that $t = t'$. By design of the algorithm, $t \geq t(u) + 2d(u,v) = t'$. Moreover, we have $t' \geq t(u) + d(u,v)$. Therefore using \cref{lemma:arrival-time}, both $u$ and $v$ have arrived at time $t'$. We now consider the state of the algorithm at time $t'$. If there were a request $w$ closer to $u$ than $v$, the algorithm would immediately match $u$ and $w$ together as we would have $t' > t(u) + 2d(u,w)$. This contradicts the fact that $u$ and $w$ got matched together. So at time $t' = t(u) + 2d(u,v)$, $u$ and $v$ have both arrived and $v$ is the closest request to $u$, so the algorithm matches them together at that time, so $t = t'$. 
    
    Thus, the time cost spent by $u$ and $v$ for waiting is $t - t(u) + t - t(v) = 4d(u,v) + t(u) - t(v) \leq 4 d(u,v)$ because $t(u) \leq t(v)$. So, the total time cost of the algorithm $\mathcal{W}_{GREEDY}$ can be bounded by:
    \begin{align*}
        \mathcal{W}_{GREEDY} &= \sum_{(u,v) \in M} 4d(u,v)\\
        & = 4 w(M)
    \end{align*}

    Therefore, the total online cost of the greedy algorithm is at most:
    \begin{align*}
        C_{GREEDY} &= w(M) + \mathcal{W}_{GREEDY} \\
        & \leq w(M) + 4 \cdot w(M) = 5 \cdot w(M) \\
        & \leq\frac{5}{2} (\lceil \log n\rceil + 1) OPT_{TSP}
    \end{align*}
\end{proof}

\end{document}